\documentclass[preprint,12pt]{elsarticle}

\usepackage{amssymb}
\usepackage{amsmath}
\usepackage{amsthm}
\usepackage{indentfirst}
\usepackage{subcaption}
\usepackage{latexsym,paralist,graphicx}
\usepackage{enumerate}
\usepackage{multirow}
\usepackage{comment}
\usepackage{algorithm}
\usepackage{algorithmic}
\usepackage{appendix}
\usepackage{xcolor}
\usepackage{graphicx}
\usepackage{epstopdf}
\usepackage{bm}
\usepackage[
    colorlinks=true, 
    allcolors=black, 
    pdfborder={0 0 0}
]{hyperref}
\usepackage[normalem]{ulem}
\usepackage{bbm}
\usepackage{mathrsfs}
\usepackage{booktabs}
\usepackage{threeparttable}
\usepackage{slashed}
\newtheorem{theorem}{Theorem}[section]
\newtheorem{lemma}{Lemma}[section]

\newtheorem{proposition}{Proposition}[section]

\newtheorem{remark}{Remark}[section]
\newtheorem{definition}{Definition}[section]
\makeatletter
\renewenvironment{proof}[1][\proofname]{\par
  \pushQED{\qed}%
  \normalfont \topsep6\p@\@plus6\p@\relax
  \trivlist
  \item[\hskip\labelsep
        \normalfont\bfseries #1\@addpunct{.}]\ignorespaces
}{%
  \popQED\endtrivlist\@endpefalse
}
\makeatother
\journal{Journal of Computational and Applied Mathematic}

\begin{document}

\begin{frontmatter}



\title{Option pricing model under the G-expectation framework}


\author[1]{Ziting Pei} %% Author name
\author[2,3]{Xingye Yue}
\author[2]{Xiaotao Zheng\corref{cor1}}
\ead{20234013002@stu.suda.edu.cn}
\cortext[cor1]{Corresponding author}
\affiliation[1]{organization={School of Business},%Department and Organization
            addressline={Suzhou University of Science and Technology},
            city={Suzhou},
            postcode={215006},
            state={Jiangsu},
            country={China}}
\affiliation[2]{organization={Center for Financial Engineering},%Department and Organization
            addressline={Soochow University},
            city={Suzhou},
            postcode={215006},
            state={Jiangsu},
            country={China}}
\affiliation[3]{organization={School of Mathematical Sciences},%Department and Organization
            addressline={Soochow University},
            city={Suzhou},
            postcode={215006},
            state={Jiangsu},
            country={China}}
%% Abstract
\begin{abstract}
G-expectation, as a sublinear expectation, provides a powerful framework for modeling uncertainty in financial markets. Motivated by the need for robust valuation under model uncertainty, this work develops a unified risk‑neutral valuation approach within the G‑expectation environment, yielding a nonlinear generalization of the Black–Scholes model, termed the G‑Black–Scholes equation. To enhance computational efficiency and reduce numerical cost, we introduce a logarithmic transformation of the asset price, which yields an alternative nonlinear PDE. Based on this transformed formulation, we design both explicit and implicit finite difference schemes that are rigorously demonstrated to be consistent, stable, monotone, and convergent to the viscosity solution. Numerical examples further confirm the accuracy and computational advantages of the proposed methods. 
Numerical examples further verify that the proposed schemes achieve high accuracy and that the logarithmic transformation relax the stability constraints of explicit schemes and improves computational efficiency.

% Numerical experiments on benchmark options, including butterfly spreads and digital options, further confirm the accuracy and computational benefits of the proposed methods.
\end{abstract}


\begin{highlights}
\item A comprehensive risk‑neutral option valuation framework is formulated within the G‑expectation setting, in which the underlying follows a G‑geometric Brownian motion. The resulting pricing problem is governed by a fully nonlinear partial differential equation.
\item Both explicit and implicit finite‑difference discretizations are constructed for the logarithmically transformed equation. Their consistency, stability, monotonicity, and convergence towards the viscosity solution are rigorously established.
\item Theoretical analysis and numerical experiments demonstrate that the logarithmic transformation improves efficiency and reduces computational cost.
\end{highlights}

%% Keywords
\end{frontmatter}


\section{Introduction}
\label{sec1}
%Derivative pricing is an important issue in financial theory and practice. The research on how to price options has a long history in finance, which can be traced back to the work of Bachelier in 1900. For a long time afterwards, the development of financial pricing theory was slow. It was not until 1973 that Black and Scholes \cite{BS} made significant contributions to option pricing theory. 

Derivative pricing is an important issue in financial theory and practice. The research on how to price options has a long history in finance, which can be traced back to the work of Bachelier in 1900. The evolution of derivative pricing models remained gradual until the breakthrough contributions of Black and Scholes \cite{BS} in 1973, which revolutionized option pricing methodology.
Under the no-arbitrage assumption, they used the dynamic replication method to obtain the explicit solution for European call options on stocks that do not pay dividends.
Subsequently, in 1997, it was honored with the Nobel Prize in Economics.

The Black-Scholes pricing formula has been widely recognized by the academic and industrial communities for its elegant form and simple calculation. However, this pricing theory itself involves some assumptions that do not align with practices. For example, one of the important assumptions of the Black-Scholes pricing theory is that stock prices adhere to a geometric Brownian motion, with volatility remaining constant. However, practical studies have shown that volatility is not constant and is more likely to be uncertain. More complex models postulate volatility surfaces that span different underlying asset prices and time periods (refer to, for example, Andersen and Brothertonratcliffe \cite{AB}; Coleman et al. \cite{Co}). Market practitioners typically construct these surfaces by applying implied volatilities calculated via the Black-Scholes model across a broad spectrum of actively traded instruments. An alternative methodology employs stochastic volatility, which posits that volatility itself follows a random walk pattern, as initially proposed by Heston \cite{He}. Stochastic volatility complicates numerical pricing methods by increasing the number of state variables required.

Lyons \cite{Ly} and Avellaneda et al. \cite{Ave} independently proposed the uncertain volatility framework. In this approach, it is assumed that the volatility of the underlying asset lies within a specified range of values.
As a result, the prices derived from a no-arbitrage analysis are no longer singular. The only calculable elements are the best-case and worst-case prices for a particular long or short position. When an investor takes the worst-case scenario into account, they can protect their position. Even if the actual volatility varies, as long as it remains within the defined range, the investor can ensure a balance that is not negative within the hedging portfolio.  
%The pricing equations derived from uncertain volatility models in finance are often cast in the form of nonlinear partial differential equations(PDEs). 
By the Feynman–Kac formula, the mathematical formulation of the uncertain volatility model naturally leads to a nonlinear partial differential equation (PDE).
For the single-factor uncertain volatility model, Pooley et al. \cite{Po} developed numerical algorithms and discussed their convergence properties.
For the two-factor case, Ma and Forsyth \cite{MF} proposed a fully implicit finite difference scheme that is unconditionally monotone.

Peng \cite{Peng1} proposed the nonlinear G-expectation and presented concepts such as the G-Brownian motion. He established the stochastic integral using G-Brownian motion and obtained the corresponding G-Itô formula. The central limit theorem in the G-framework derived by Peng \cite{Peng2} provides a theoretical basis for describing the price changes of marketable securities with
the G-geometric Brownian motion. The G-geometric Brownian motion can be used to characterize the part of stochastic uncertainty.

This paper studies the risk-neutral pricing problem of options under the G-expectation. In this framework, the G-expected returns of both the underlying assets and the options are the risk-free interest rates. The price $S_{t}$ of the underlying asset is described by the G-geometric Brownian motion, and the price of the option corresponds to a G-expectation. Using the nonlinear Feynman-Kac formula, the option pricing problem is transformed into solving a fully nonlinear PDE. 
Furthermore, to improve numerical performance, a logarithmic transformation is applied to the price of the underlying asset, and $X_{t}=\mathrm{ln}S_{t}$ is introduced. 
This transformation significantly relaxes the stability  constraint and reduces the computational cost of the explicit scheme.
Based on the equation of $X_{t}$, this paper presents both explicit and implicit numerical schemes for solving the fully nonlinear PDE respectively. The consistency, stability, and monotonicity of the proposed schemes are rigorously established, thereby guaranteeing convergence to the viscosity solution of the nonlinear PDE. Numerical examples further confirm the accuracy and computational advantages of the proposed methods.

The remainder of this paper is organized as follows. Section \ref{sec:3} establishes the risk-neutral pricing framework under G-expectation, derives the G-Black-Scholes equation, and introduces a logarithmic transformation to obtain an equivalent nonlinear PDE with improved computational properties. Section \ref{sec:4} presents explicit and implicit finite difference schemes for solving the transformed PDE. Section \ref{sec:5} provides rigorous theoretical analysis, proving consistency, stability, monotonicity, and convergence to the viscosity solution, with convergence rate estimates and proof of convergence for the implicit scheme's nonlinear iteration. Section \ref{sec:6} presents numerical experiments confirming the accuracy and computational advantages of the proposed methods. Section \ref{sec:7} concludes the paper.

\section{Risk-neutral option pricing under G-expectation} {\label{sec:3}}
In this section, we establish a risk-neutral option pricing framework under G-expectation. We begin by defining risk neutrality in the nonlinear expectation setting, then construct a G-option pricing model and derive the corresponding G-Black-Scholes equation. Finally, we introduce a logarithmic transformation to obtain an alternative formulation that improves computational efficiency.

Consider a financial market consisting of three fundamental instruments: risk-free bonds, underlying assets (such as stocks), and derivatives (such as options) written on these assets. The risk-free bond grows at the constant rate $r > 0$, while the underlying asset and derivatives are subject to market uncertainty.
In the classical linear probability framework, risk neutrality is a well-established concept where all securities in the market earn the risk-free rate of interest $r$ under an appropriately chosen risk-neutral probability measure. This fundamental principle forms the cornerstone of modern derivative pricing theory, ensuring that arbitrage opportunities are eliminated and enabling the unique valuation of contingent claims.

However, in the nonlinear expectation framework, there exists no concept of probability measure, only the nonlinear expectation operator $\mathbb{E}$. Therefore, we extend the notion of risk neutrality to the G-expectation setting as follows:

\begin{definition}[Risk-neutral pricing under G-expectation]
\label{def:risk_neutral_G}
In a G-expectation framework where the market is risk-neutral, the G-expected return on all risky securities equals the risk-free rate of interest $r$. For an underlying asset price process $S_t$ and a derivative process $U(t,S_t)$ in the G-expectation space $(\Omega, \mathcal{H}, \mathbb{E})$, this implies
\begin{align}
\mathbb{E}\left[\frac{S_T}{\Gamma_T} \right] = \frac{S_t}{\Gamma_t}, \quad 
\mathbb{E}\left[\frac{U(T, S_T)}{\Gamma_T} \right] = \frac{U(t, S_t)}{\Gamma_t},
\end{align}
where $\Gamma_t = e^{rt}$ is the risk-free bond. Consequently, under this risk-neutral market, the fair value of the derivative at time $t = 0$ with initial asset price $S_0$ is given by
\begin{align}
U(0, S_0) = e^{-rT} \mathbb{E}[U(T,S_T)].
\end{align}

\end{definition}

Previous literature has explored option pricing in incomplete markets under the G-expectation framework through hedging strategies \cite{Denis_06, Xu_2021, Xu_2023}.
Building upon this foundation, we now formulate a complete risk-neutral pricing framework under the G-expectation.

A fundamental prerequisite for option pricing is establishing the risk-neutral dynamics of the underlying asset. Consider a stock price process $S(t)$ driven by G-geometric Brownian motion in a G-expectation space $(\Omega, \mathcal{H}, \mathbb{E})$ with a symmetric G-Brownian motion $B_t$:
\begin{equation} 
\label{G_St}
S_t = S_0 e^{rt - \frac{\sigma^2}{2}\langle B \rangle_t + \sigma B_t},
\end{equation}
where $S_0 > 0$ is the initial stock price, $r \geq 0$ is the risk-free interest rate, $\sigma > 0$ is the volatility parameter, $\langle B \rangle_t$ denotes the G-quadratic variation of $B_t$, and $B_t \sim \mathcal{N}(0,[\underline{\Sigma}^2,\overline{\Sigma}^2])$.
By Lemma 3.3.4 and Lemma 3.6.1 of \cite{Peng2}, it can be shown that the stock price process $S(t)$ is risk-neutral under G-expectation and satisfies:
\begin{equation} 
\label{eq:dSt}
\left\{ \begin{array}{l}
dS_t = rS_t dt + \sigma S_t dB_t, \\
S_0 = s_0,
\end{array} \right.
\end{equation}
and
\begin{equation} 
\mathbb{E}[S(t)]=S_0e^{rt}.
\end{equation}
Given the risk-neutral stochastic process $S_t$ \eqref{eq:dSt} under the G-expectation framework, we define the value function $U(0,s_0)$ of the G-option as:
\begin{equation} \label{value_function}
	\begin{aligned}
	U\left( 0,s_0\right) &=\mathbb{E}\left[ \left. e^{-r\left( T-t\right) }\phi
	\left( S_{T}\right) \right\vert S_{0}=s_0\right] \\
        &=e^{-r\left( T-t\right) } \mathbb{E}\left[ \left. \phi
	\left( S_{T}\right) \right\vert S_{0}=s_0\right],
\end{aligned}
\end{equation}
where $\phi \left( S_{T}\right) $ is the payoff function.

\subsection{G-Black-Scholes equation}

By the nonlinear Feynman-Kac formula of \cite{Peng2}, the value function
$U(t,S)$ satisfies the following fully nonlinear PDE:
\begin{equation} \label{G_S_option_eq}
\left\{
\begin{aligned}
&\frac{\partial U}{\partial t} + rS \frac{\partial U}{\partial S}
+ \frac{1}{2} \sup_{\Sigma \in [\underline{\Sigma}, \overline{\Sigma}]} \left(
\sigma^2 S^2 \Sigma^2 \frac{\partial^2 U}{\partial S^2} \right)
- rU = 0, \quad 0 \leq t \leq T, \\
&U(t, S)\big|_{t = T} = \phi(S).
\end{aligned}
\right.
\end{equation}

\begin{remark}
When $\phi(S) = \max(S-K,0)$ (convex) or $\phi(S) = \max(K-S,0)$ (concave), the equation reduces to the classical Black-Scholes equation. Specifically, for a call option with convex payoff $\phi(S) = \max(S-K,0)$, equation \eqref{G_S_option_eq} becomes:
\begin{equation}  \label{App_eq:BS_SUP}
\left\{
\begin{aligned}
&\frac{\partial U}{\partial t} + rS \frac{\partial U}{\partial S}
+ \frac{1}{2} \left(
\sigma^2 S^2 \overline{\Sigma}^2 \frac{\partial^2 U}{\partial S^2} \right)
- rU = 0, \quad 0 \leq t \leq T, \\
&U(t, S)\big|_{t = T} = max(S-K,0),
\end{aligned}
\right.
\end{equation}
while for a put option with concave payoff $\phi(S) = \max(K-S,0)$, the equation becomes:
\begin{equation}
\left\{
\begin{aligned}
&\frac{\partial U}{\partial t} + rS \frac{\partial U}{\partial S}
+ \frac{1}{2} \left(
\sigma^2 S^2 \underline{\Sigma}^2 \frac{\partial^2 U}{\partial S^2} \right)
- rU = 0, \quad 0 \leq t \leq T, \\
&U(t, S)\big|_{t = T} = max(K-S,0).
\end{aligned}
\right.
\end{equation}
% The proof can be found in the appendix \ref{Appendix:Property}.
\end{remark}

\begin{remark}
    Clearly, when the symmetric G-Brownian motion degenerates to a standard Brownian motion, the uncertainty in $\Sigma$ vanishes, and the equation \eqref{G_S_option_eq} reduces to the classical Black-Scholes equation.
    % we can solve the eqaution by FDM (refer to \cite{}).
\end{remark}

\subsection{Logarithmic Transformation and Alternative Formulation}

After establishing the G-option pricing model under the asset price process \( S_t \) and deriving the associated fully nonlinear PDE, to improve computational efficiency, we introduce a logarithmic transformation on the underlying asset by setting \( X = \ln S \). The convenience of this calculation is mainly reflected in the explicit difference method. We will elaborate on the advantages of the logarithmic transformation in the numerical analysis of the explicit method in Section \ref{sec:5}.

Under this transformation, we formulate an alternative nonlinear PDE given by,
\begin{equation} \label{G-option_eq}
    \left\{
    \begin{aligned}
    &\frac{\partial V}{\partial t} + r\frac{\partial V}{\partial X}
    + \sup_{\Sigma \in [\underline{\Sigma}, \overline{\Sigma}]} \left[
    \frac{1}{2} \sigma^2 \Sigma^2 \cdot \left(\frac{\partial^2 V}{\partial X^2} - \frac{\partial V}{\partial X} \right) \right]
    - rV = 0, \quad 0 \leq t \leq T, \\
    &V(t, X)\big|_{t = T} = \phi(e^{X}).
    \end{aligned}
    \right.
\end{equation}
%This paper propose the FDM for this equation \eqref{G-option_eq}.

\begin{remark}
    When the symmetric G-Brownian motion degenerates into a classical Brownian motion, equation \eqref{G-option_eq} reduces to the classical result.
    Moreover, we consider the log-price process $X_t$,
    \begin{equation}\label{SOCM_X}
        \left\{ \begin{array}{l}
        d{X_t} = rdt - \frac{\sigma^2}{2}  d\langle B \rangle_t + \sigma dB_t,  \\
        X_0 = ln(s_0),
    \end{array} \right.
    \end{equation}
    the viscosity solution to the nonlinear PDE \eqref{G-option_eq} corresponds to the value function of the following stochastic problem,
    \begin{equation}\label{SOCM_V}
        V\left( t,x\right) = \mathbb{E}\left[ \left. e^{-r\left( T-t\right) }\phi \left( e^{X_{T}}\right) \right\vert X_{t}=x\right].
    \end{equation}
\end{remark}

In the following Section \ref{sec:4}, we will present explicit and implicit difference methods for solving the nonlinear PDE \eqref{G-option_eq} respectively.

\section{Finite difference methods}
%We have demonstrated that, when asset prices follow G-geometric Brownian
%motion, the risk-neutral option prices—under the framework of nonlinear expectations—satisfy both the stochastic optimal control models (SOCMs) and the corresponding nonlinear Black–Scholes-type equations \eqref{G_S_option_eq} and \eqref{G-option_eq}. In this section, we focus on developing an efficient finite difference method (FDM) to numerically solve the nonlinear partial differential equation \eqref{G-option_eq}, which arises from the logarithmic transformation of the asset price process.
{\label{sec:4}}
% We have demonstrated that, when asset prices follow G-geometric Brownian motion, the risk-neutral option prices under the framework of nonlinear expectations satisfy the following SOCMs and nonlinear Black-Scholes equations.
Without loss of generality (WLOG),  we set the volatility parameter \( \sigma = 1 \).
 Under this assumption, the asset price dynamics given in \eqref{G_St} can be rewritten as \begin{equation}
S_{t}=S_{0}e^{rt-\frac{1}{2}\left\langle B\right\rangle _{t}+B_{t}}.
\end{equation}
Then PDE \eqref{G_S_option_eq} becomes
% where $B_{t}\overset{d}{=}\mathcal{N}\left( 0, \left[\underline{\Sigma}^{2}, \overline{\Sigma}^{2}  \right]\right) $, denotes a G-Brownian motion with uncertain volatility in the interval $\left[\underline{\Sigma}^{2}, \overline{\Sigma}^{2}  \right]$.
% Let  $%
% \phi \left( S_{T}\right) $ be the payoff function, then the option value is given by
% \begin{equation*}
% U\left( t,S\right) =\mathbb{E}\left[ \left. e^{-r\left( T-t\right) }\phi
% \left( S_{T}\right) \right\vert S_{t}=s\right] .
% \end{equation*}%
% The corresponding nonlinear PDE for \( U(t, S) \) takes the form,
\begin{equation}
 \left\{
\begin{array}{l}
 U_{t}+rSU_{S}+\frac{1}{2}\underset{\Sigma\in \left[ \underline{\Sigma }%
,\overline{\Sigma }\right] }{\sup }\Sigma ^{2}S^{2}U_{SS}-rU=0, \\
U_{T}=\phi \left( S\right) .
 \end{array}
 \right.  \label{G-op1}
\end{equation}
% % At $s=0$, we have the boundary condition \begin{equation}
% % U_{t}=rU,
% % \label{ss}
% % \end{equation}
% % while at
Moreover, we make a substitution of independent variable $t=T-t$ to transform PDE  \eqref{G-option_eq}  into the Cauchy problem of a fully nonlinear parabolic equation. The PDE \eqref{G-option_eq}
can be rewritten as
\begin{equation}
\left\{
\begin{array}{l}
V_{t}-rV_{X}+rV-\frac{1}{2}\underset{\Sigma \in \left[ \underline{\Sigma
},\overline{\Sigma }\right] }{\sup }\Sigma ^{2} \cdot \left(
V_{XX}-V_{X}\right) =0, \\
V_{0}=\phi \left( e^{X}\right) .%
\end{array}%
\right.  \label{G-option11}
\end{equation}%
At $X\rightarrow -\infty ,$ we have the boundary condition
\begin{equation}
V_{t}=-rV.
\label{xs}
\end{equation}
% while at $x\rightarrow \infty $, we have a Dirichlet condition
% \begin{equation}
% V\simeq A(t)e^{x}+B(t),
% \label{xl}
% \end{equation}%
% where $A$ and $B$ can be determined by financial reasoning.
It is easy to see that the optimal $\Sigma ^{*}$ is reached at the end of
the intervals, depending only on the sign of the difference between the
second-order derivative and the first-order derivative, that is,
\begin{equation}
\Sigma ^{*} = \Sigma \left( V_{XX}-V_{X}\right) =\left\{
\begin{array}{l}
\overline{\Sigma }, \quad\text{ \ \ }V_{XX}-V_{X}\geq 0, \\
\underline{\Sigma }, \quad \text{ \ \ }V_{XX}-V_{X}<0.%
\end{array}%
\right.   \label{b12}
\end{equation}%
% We will illustrate the advantages of applying such a logarithmic
% transformation in the following numerical analysis.

\subsection{Discretization}
For computational purpose, we confine the problem \eqref{G-option11} within a truncated bounded domain,
\begin{equation*}
0\leq t\leq T\text{ and }\left\vert X\right\vert <L,
\end{equation*}%
with artificial boundary condition at the right boundary. For different
contracts, the artificial boundary conditions are modified accordingly.
Subsequently, the problem is reformulated as

\begin{equation}
\left\{
\begin{array}{l}
V_{t}-rV_{X}+rV-\frac{\Sigma ^{2}\left( V_{XX}-V_{X}\right) }{2}\left(
V_{XX}-V_{X}\right) =0,\quad t\in (0,T),\text{ }X\in (-L,L), \\
\left. V\right\vert _{t=0}=\phi (e^{X}), \\
\left. V_{t}\right\vert _{X=-L}=-rV, \\
\left. V\right\vert _{X=L}=\varphi(t) .%
\end{array}%
\right.   \label{G-option1}
\end{equation}

Taking an equi-distance partition with a spatial step size $h=2L/M,\Delta
t=T/N$, we have grid points $X_{i}=-L+i\ast h,t^{n}=n\Delta t,$ for $%
i=0,\cdots ,M,$ and $n=0,\cdots ,N$. Let $V_{i}^{n}$ denote the approximate solution at $\left(
t^{n},X_{i}\right) $. In this paper, numerical methods are designed to solve PDE \eqref{G-option1}.
\subsection{Explicit difference method}
An explicit scheme for equation \eqref{G-option1}
reads as, for  $n=0,\cdots ,N-1$,
\begin{equation}
\left\{
\begin{array}{l}
\delta _{t}V_{i}^{n+1}-r\delta _{h}V_{i}^{n}+rV_{i}^{n+1}-\frac{\Sigma
^{2}\left( \delta _{h}^{2}V_{i}^{n}-\delta _{h}V_{i}^{n}\right) }{2}\left(
\delta _{h}^{2}V_{i}^{n}-\delta _{h}V_{i}^{n}\right) =0,0<i<M, \\
V_{i}^{0}=\phi \left( e^{X_{i}}\right) ,i=0,...,M, \\
\left. \delta _{t}V_{i}^{n+1}\right\vert _{X_{i}=-L}=-rV_{i}^{n+1}, \\
\left. V_{i}^{n+1}\right\vert _{X_{i}=L}=\varphi ^{n+1},%
\end{array}%
\right. \label{explicit1}
\end{equation}
where
\begin{eqnarray}
\delta _{t}V_{i}^{n+1} &=&\frac{V_{i}^{n+1}-V_{i}^{n}}{\Delta t},\text{ }%
\delta _{h}V_{i}^{n}=\frac{V_{i+1}^{n}-V_{i-1}^{n}}{2h},\text{ \ }
\label{sign1} \\
\delta _{h}^{2}V_{i}^{n} &=&\frac{V_{i+1}^{n}-2V_{i}^{n}+V_{i-1}^{n}}{h^{2}%
}.  \label{sign2}
\end{eqnarray}%

\subsection{Implicit difference method}
An implicit scheme for equation \eqref{G-option1}
reads as, for $n=0,\cdots ,N-1$,
\begin{equation}
\left\{
\begin{array}{l}
\delta _{t}V_{i}^{n+1}-r\delta _{h}V_{i}^{n+1}+rV_{i}^{n+1}-\frac{\Sigma
^{2}\left( \delta _{h}^{2}V_{i}^{n+1}-\delta _{h}V_{i}^{n+1}\right) }{2}%
\left( \delta _{h}^{2}V_{i}^{n+1}-\delta _{h}V_{i}^{n+1}\right) =0,0<i<M, \\
V_{i}^{0}=\phi \left( e^{X_{i}}\right) ,i=0,...,M, \\
\left. \delta _{t}V_{i}^{n+1}\right\vert _{X_{i}=-L}=-rV_{i}^{n+1}, \\
\left. V_{i}^{n+1}\right\vert _{X_{i}=L}=\varphi ^{n+1},%
\end{array}%
\right.   \label{implicit1}
\end{equation}
where $\delta _{t}V_{i}^{n+1}$, $\delta _{h}V_{i}^{n+1}$, $\delta _{h}^{2}V_{i}^{n+1}$ take the same forms as in equations \eqref{sign1}-\eqref{sign2}. Since \eqref{implicit1} is a nonlinear system, an inner iteration is needed to obtain the solution $V_{i}^{n+1}$ in each time step. Let $V_{i}^{n+1,k}$ be the $k^{th}$ estimate for $V_{i}^{n+1}$, $V_{i}^{n+1,k}$ is given by the following Picard's iteration,
\begin{equation}
\left\{
\begin{aligned}
&\delta _{t}V_{i}^{n+1,k+1}-r\delta _{h}V_{i}^{n+1,k+1}+rV_{i}^{n+1,k+1} \\
&\quad -\frac{\Sigma ^{2}\left( \delta _{h}^{2}V_{i}^{n+1,k}-\delta
_{h}V_{i}^{n+1,k}\right) }{2}\left( \delta _{h}^{2}V_{i}^{n+1,k+1}-\delta
_{h}V_{i}^{n+1,k+1}\right) =0, \\
&V_{i}^{0}=\phi \left( e^{X_{i}}\right) , \\
&\left. \delta _{t}V_{i}^{n+1,k+1}\right\vert _{X_{i}=-L}=-rV_{i}^{n+1,k+1}, \\
&\left. V_{i}^{n+1,k+1}\right\vert _{X_{i}=L}=\varphi ^{n+1}.
\end{aligned}
\right.   \label{iteration}
\end{equation}
\begin{comment}
\begin{equation}
\left\{
\begin{array}{l}
\delta _{t}V_{i}^{n+1,k+1}-r\delta _{h}V_{i}^{n+1,k+1}+rV_{i}^{n+1,k+1}-%
\frac{\Sigma ^{2}\left( \delta _{h}^{2}V_{i}^{n+1,k}-\delta
_{h}V_{i}^{n+1,k}\right) }{2}\left( \delta _{h}^{2}V_{i}^{n+1,k+1}-\delta
_{h}V_{i}^{n+1,k+1}\right) =0, \\
V_{i}^{0}=\phi \left( e^{X_{i}}\right) , \\
\left. \delta _{t}V_{i}^{n+1,k+1}\right\vert _{X_{i}=-L}=-rV_{i}^{n+1,k+1},
\\
\left. V_{i}^{n+1,k+1}\right\vert _{X_{i}=L}=\varphi ^{n+1}.%
\end{array}%
\right.   \label{iteration}
\end{equation}
\end{comment}

\section{Numerical analysis for numerical methods}\label{sec:5}
From the work of Barles and Souganidis \cite{BS2}, we know that numerical solution of %
\eqref{explicit1} or \eqref{implicit1} converges to the viscosity solution of the PDE %
\eqref{G-option1} if the method is consistent, stable (in the $l_{\infty }$
norm) and monotone.

We now give the definition of monotonicity. Denote by
\begin{equation*}
g_{i}=g\left( V_{i}^{n+1},\{V_{k}^{n+1}\}_{k\in
N_{i}},V_{i}^{n},\{V_{k}^{n}\}_{k\in N_{i}}\right)
\end{equation*}%
the left-hand side of the difference equation \eqref{explicit1} or %
\eqref{implicit1}. Here, $N_{i}=\{i+1,i-1\}$ presents the set of all
nearest-neighbor indexes of $i$.

\begin{definition}
\textrm{(Monotonicity)} \label{def-mono} The scheme \eqref{explicit1} or \eqref{implicit1} is
monotone if for all $i$,
\begin{equation}
\begin{aligned}
&g_{i}\left( V_{i}^{n+1}+\epsilon _{i}^{n+1},\{V_{k}^{n+1}\}_{k\in
N_{i}},V_{i}^{n},\{V_{k}^{n}\}_{k\in N_{i}}\right) \\
\geq &g_{i}\left(
V_{i}^{n+1},\{V_{k}^{n+1}\}_{k\in N_{i}},V_{i}^{n},\{V_{k}^{n}\}_{k\in
N_{i}}\right) , \quad
 \epsilon _{i}^{n+1}\geq 0,%
\end{aligned}
\label{g-diag}
\end{equation}%
and%
\begin{equation}
\begin{aligned}
&g_{i}\left( V_{i}^{n+1},\{V_{k}^{n+1}+\epsilon _{k}^{n+1}\}_{k\in
N_{i}},V_{i}^{n}+\epsilon _{i}^{n},\{V_{k}^{n}+\epsilon _{k}^{n}\}_{k\in
N_{i}}\right) \\
\leq &g_{i}\left( V_{i}^{n+1},\{V_{k}^{n+1}\}_{k\in
N_{i}},V_{i}^{n},\{V_{k}^{n}\}_{k\in N_{i}}\right) , 
\quad \epsilon _{i}^{n},\epsilon _{k}^{n+1},\epsilon _{k}^{n}\geq 0.%
\end{aligned}
\label{g-off-diag}
\end{equation}
\end{definition}% \begin{definition}
% \textrm{(Monotonicity)} \label{def-mono} The scheme \eqref{explicit1} is
% monotone if for all $i$,
% \begin{equation}
% g_{i}\left( V_{i}^{n+1}+\epsilon _{i}^{n+1},\{V_{k}^{n+1}\}_{k\in
% N_{i}},V_{i}^{n},\{V_{k}^{n}\}_{k\in N_{i}}\right) \geq g_{i}\left(
% V_{i}^{n+1},\{V_{k}^{n+1}\}_{k\in N_{i}},V_{i}^{n},\{V_{k}^{n}\}_{k\in
% N_{i}}\right) ,\text{ }\forall \epsilon _{i}^{n+1}\geq 0,  \label{g-diag}
% \end{equation}%
% and
% \begin{equation}
% g_{i}\left( V_{i}^{n+1},\{V_{k}^{n+1}+\epsilon _{k}^{n+1}\}_{k\in
% N_{i}},V_{i}^{n}+\epsilon _{i}^{n},\{V_{k}^{n}+\epsilon _{k}^{n}\}_{k\in
% N_{i}}\right) \leq g_{i}\left( V_{i}^{n+1},\{V_{k}^{n+1}\}_{k\in
% N_{i}},V_{i}^{n},\{V_{k}^{n}\}_{k\in N_{i}}\right) ,\text{ }\forall \epsilon
% _{i}^{n},\epsilon _{k}^{n+1},\epsilon _{k}^{n}\geq 0.  \label{g-off-diag}
% \end{equation}
% \end{definition}
\subsection{Monotonicity and Convergence of explicit scheme \texorpdfstring{\eqref{explicit1}}{explicit scheme}}
% \begin{assumption}
% \label{assume1}The constraint on the mesh ratio is $\overline{\Sigma }\sqrt{%
% \Delta t}\leq h\leq \frac{2\underline{\Sigma }^{2}}{\max \left( 2r-%
% \underline{\Sigma }^{2},\overline{\Sigma }^{2}-2r\right) }$.
% \end{assumption}

\begin{lemma}
\label{Consistency} \textrm{(Consistency) }The explicit scheme %
\eqref{explicit1} is consistent.
\end{lemma}

\begin{proof}
Consistency is verified by observing that the equation in \eqref{explicit1} tends
to the PDE \eqref{G-option1} as $h,\Delta t\rightarrow 0$, since the `$\sup $%
' operation is continuous.
\end{proof}

% We now give the definition of monotonicity. Denote by
% \begin{equation*}
% g_{i}=g\left( V_{i}^{n+1},V_{i}^{n},\{V_{k}^{n}\}_{k\in N_{i}}\right)
% \end{equation*}%
% the left-hand side of the difference equation \eqref{explicit1}. Here, $%
% N_{i}=\{i+1,i-1\}$ presents the set of all nearest-neighbor indexes of $i$.

% \begin{definition}
% \textrm{(Monotonicity)} \label{def-mono} The scheme \eqref{explicit1} is
% monotone if for all $i$,
% \begin{equation}
% g_{i}\left( V_{i}^{n+1}+\epsilon _{i}^{n+1},V_{i}^{n},\{V_{k}^{n}\}_{k\in
% N_{i}}\right) \geq g_{i}\left( V_{i}^{n+1},V_{i}^{n},\{V_{k}^{n}\}_{k\in
% N_{i}}\right) ,\text{ }\forall \epsilon _{i}^{n+1}\geq 0,  \label{g-diag}
% \end{equation}%
% and
% \begin{equation}
% g_{i}\left( V_{i}^{n+1},V_{i}^{n}+\epsilon _{i}^{n},\{V_{k}^{n}+\epsilon
% _{k}^{n}\}_{k\in N_{i}}\right) \leq g_{i}\left(
% V_{i}^{n+1},V_{i}^{n},\{V_{k}^{n}\}_{k\in N_{i}}\right) ,\text{ }\forall
% \epsilon _{i}^{n},\epsilon _{k}^{n}\geq 0.  \label{g-off-diag}
% \end{equation}
% \end{definition}

\begin{lemma}
\label{mono11}
\textrm{(Monotonicity)} If condition $\overline{\Sigma }\sqrt{\Delta t}\leq h\leq \frac{2\underline{\Sigma }^{2}}{\max \left( 2r-\underline{\Sigma }^{2},\overline{\Sigma }^{2}-2r\right) }$ is satisfied, then the explicit
scheme \eqref{explicit1} is monotone.
\end{lemma}

\begin{proof}
We first consider the perturbation on $V_{i}^{n+1}$, it is obvious from
equation \eqref{explicit1} that
\begin{equation}
g_{i}\left( V_{i}^{n+1}+\epsilon _{i}^{n+1},V_{i}^{n},\{V_{k}^{n}\}_{k\in
N_{i}}\right) -g_{i}\left( V_{i}^{n+1},V_{i}^{n},\{V_{k}^{n}\}_{k\in
N_{i}}\right) \geq 0.  \label{T-diag}
\end{equation}%
We now turn to the perturbation on $V_{i}^{n}$ and $\{V_{k}^{n}\}_{k\in
N_{i}}.$ Denote by $\widetilde{V}_{i}^{n}=V_{i}^{n}+$ $\epsilon _{i}^{n},$
for $\epsilon _{i}^{n}\geq 0.$ We also denote by $\widetilde{V}%
_{k}^{n}=V_{k}^{n}+\epsilon _{k}^{n}$, for $\epsilon _{k}^{n}\geq 0$ and $%
k\in N_{i}$. Then the difference between the two sides of the inequality %
\eqref{g-off-diag} is
 \begin{eqnarray*}
T: &=&g_{i}\left( V_{i}^{n+1},\widetilde{V}_{i}^{n},\{\widetilde{V}%
_{k}^{n}\}_{k\in N_{i}}\right) -g_{i}\left(
V_{i}^{n+1},V_{i}^{n},\{V_{k}^{n}\}_{k\in N_{i}}\right)   \notag \\
&=&-\frac{\epsilon _{i}^{n}}{\Delta t}-r\delta _{h}\widetilde{V}%
_{i}^{n}+r\delta _{h}V_{i}^{n}-\frac{\widetilde{\Sigma }^{2}}{2} \cdot \left(
\delta _{h}^{2}\widetilde{V}_{i}^{n}-\delta _{h}\widetilde{V}_{i}^{n}\right)
+\frac{\widehat{\Sigma }^{2}}{2} \cdot\left( \delta _{h}^{2}V_{i}^{n}-\delta
_{h}V_{i}^{n}\right)   \notag \\
&=&-\frac{\epsilon _{i}^{n}}{\Delta t}-r\delta _{h}\widetilde{V}%
_{i}^{n}+r\delta _{h}V_{i}^{n}-\frac{\widetilde{\Sigma }^{2}}{2} \cdot \left(
\delta _{h}^{2}\widetilde{V}_{i}^{n}-\delta _{h}\widetilde{V}_{i}^{n}\right)
+\frac{\widehat{\Sigma }^{2}}{2} \cdot \left( \delta _{h}^{2}\widetilde{V}%
_{i}^{n}-\delta _{h}\widetilde{V}_{i}^{n}\right)   \notag \\
&&+\frac{\widehat{\Sigma }^{2}}{2} \cdot \left( \delta _{h}^{2}V_{i}^{n}-\delta
_{h}V_{i}^{n}\right) -\frac{\widehat{\Sigma }^{2}}{2} \cdot \left( \delta _{h}^{2}%
\widetilde{V}_{i}^{n}-\delta _{h}\widetilde{V}_{i}^{n}\right)   \notag \\
&\leq &-\left( \frac{1}{\Delta t}-\frac{\widehat{\Sigma }^{2}}{h^{2}}\right)
\epsilon _{i}^{n}-\left( \frac{r}{2h}+\frac{\widehat{\Sigma }^{2}}{2h^{2}}-%
\frac{\widehat{\Sigma }^{2}}{4h}\right) \epsilon _{i+1}^{n}-\left( \frac{%
\widehat{\Sigma }^{2}}{2h^{2}}-\frac{r}{2h}+\frac{\widehat{\Sigma }^{2}}{4h}%
\right) \epsilon _{i-1}^{n},  
\end{eqnarray*}%
%\end{equation*}
since $\frac{\widetilde{\Sigma }^{2}} {2} \cdot \left( \delta _{h}^{2}\widetilde{V}%
_{i}^{n}-\delta _{h}\widetilde{V}_{i}^{n}\right) \geq \frac{\widehat{\Sigma }%
^{2}}{2}\cdot\left( \delta _{h}^{2}\widetilde{V}_{i}^{n}-\delta _{h}\widetilde{V}%
_{i}^{n}\right) ,$ where
\begin{equation*}
\widetilde{\Sigma }=\Sigma\left( \delta _{h}^{2}\widetilde{V}%
_{i}^{n}-\delta _{h}\widetilde{V}_{i}^{n}\right) ,\text{ }\widehat{\Sigma }%
=\Sigma \left( \delta _{h}^{2}V_{i}^{n}-\delta _{h}V_{i}^{n}\right) .
\end{equation*}%
If condition $\overline{\Sigma }\sqrt{\Delta t}\leq h\leq \frac{2\underline{\Sigma }^{2}}{\max \left( 2r-\underline{\Sigma }^{2},\overline{\Sigma }^{2}-2r\right) }$ is satisfied, then we obtain $T\leq 0.$ The monotonicity of
scheme \eqref{explicit1} now follows directly from Definition \ref{def-mono}.
\end{proof}
\begin{lemma}
\label{Stability} \textrm{(Stability)} If condition $\overline{\Sigma }\sqrt{%
\Delta t}\leq h\leq \frac{2\underline{\Sigma }^{2}}{\max \left( 2r-%
\underline{\Sigma }^{2},\overline{\Sigma }^{2}-2r\right) }$ is satisfied,
then the explicit scheme \eqref{explicit1} is stable, in the sense that
\begin{equation}
\underset{n}{\max }\left\vert \left\vert V^{n}\right\vert \right\vert
_{\infty }\leq \max \left( \left\vert \left\vert \phi \right\vert
\right\vert _{\infty },\underset{n}{\max }\left\vert \varphi ^{n}\right\vert
\right) .  \label{uncondition2}
\end{equation}
\end{lemma}

\begin{proof}
For notational simplicity, we let $\mathbbm{I}_{h}V_{i}^{n}=\delta
_{h}^{2}V_{i}^{n}-\delta _{h}V_{i}^{n}$, then we can rewrite %
\eqref{explicit1} as the following equation
%\begin{equation}
\begin{equation}
\begin{aligned}
V_{i}^{n+1} &=\frac{1}{1+r\Delta t}\left\{ \frac{\Delta t}{2h}\left( \frac{%
\Sigma ^{2}\left( \mathbbm{I}_{h}V_{i}^{n}\right) }{h}+\frac{\Sigma
^{2}\left( \mathbbm{I}_{h}V_{i}^{n}\right) }{2}-r\right) V_{i-1}^{n}\right. \\
&\left. +\left( 1-\frac{\Sigma ^{2}\left( \mathbbm{I}_{h}V_{i}^{n}\right)
\Delta t}{h^{2}}\right) V_{i}^{n}+\frac{\Delta t}{2h}\left( r+\frac{\Sigma
^{2}\left( \mathbbm{I}_{h}V_{i}^{n}\right) }{h}-\frac{\Sigma ^{2}\left( %
\mathbbm{I}_{h}V_{i}^{n}\right) }{2}\right) V_{i+1}^{n}\right\}.
\end{aligned}
\label{explicit2}
\end{equation}

% \begin{eqnarray}
% V_{i}^{n+1} &=&\frac{1}{1+r\Delta t}\left\{ \frac{\Delta t}{2h}\left( \frac{%
% \Sigma ^{2}\left( \mathbbm{I}_{h}V_{i}^{n}\right) }{h}+\frac{\Sigma
% ^{2}\left( \mathbbm{I}_{h}V_{i}^{n}\right) }{2}-r\right) V_{i-1}^{n}+\left( 1-\frac{\Sigma ^{2}\left( \mathbbm{I}_{h}V_{i}^{n}\right)
% \Delta t}{h^{2}}\right) V_{i}^{n}\right.
% \notag \\
% &&\left. +\frac{\Delta t}{2h}\left( r+\frac{\Sigma
% ^{2}\left( \mathbbm{I}_{h}V_{i}^{n}\right) }{h}-\frac{\Sigma ^{2}\left( %
% \mathbbm{I}_{h}V_{i}^{n}\right) }{2}\right) V_{i+1}^{n}\right\} .
% \label{explicit2}
% \end{eqnarray}
%\end{equation}
Let $i_{0}$ be an index such that $\left\vert V_{i_{0}}^{n+1}\right\vert
=\left\vert \left\vert V^{n+1}\right\vert \right\vert _{\infty }$. Then it is
easy to prove that the maximum value cannot be attained at the left
boundary. However it may be attained at the right boundary, that is, $i_{0}=M,$
we have
\begin{equation}
\left\vert \left\vert V^{n+1}\right\vert \right\vert _{\infty }\leq
\left\vert \varphi ^{n+1}\right\vert .  \label{sta1}
\end{equation}%
Then for $0<i_{0}<M,$ from the discrete scheme \eqref{explicit2} and the
condition $\overline{\Sigma }\sqrt{\Delta t}\leq h\leq \frac{2\underline{%
\Sigma }^{2}}{\max \left( 2r-\underline{\Sigma }^{2},\overline{\Sigma }%
^{2}-2r\right) }$, we have
\begin{equation}
\begin{aligned}
\left\vert \left\vert V^{n+1}\right\vert \right\vert _{\infty }
&=\left\vert V_{i_{0}}^{n+1}\right\vert\\
&\leq \frac{1}{1+r\Delta t}\left\{ \frac{\Delta t}{2h}\left( \frac{\Sigma
^{2}\left( \mathbbm{I}_{h}V_{i}^{n}\right) }{h}+\frac{\Sigma ^{2}\left( %
\mathbbm{I}_{h}V_{i}^{n}\right) }{2}-r\right) \left\vert \left\vert
V^{n}\right\vert \right\vert _{\infty }\right.\\
&\left. +\left( 1-\frac{\Sigma ^{2}\left( \mathbbm{I}_{h}V_{i}^{n}\right)
\Delta t}{h^{2}}\right) \left\vert \left\vert V^{n}\right\vert \right\vert
_{\infty }+\frac{\Delta t}{2h}\left( r+\frac{\Sigma ^{2}\left( \mathbbm{I}%
_{h}V_{i}^{n}\right) }{h}-\frac{\Sigma ^{2}\left( \mathbbm{I}%
_{h}V_{i}^{n}\right) }{2}\right) \left\vert \left\vert V^{n}\right\vert
\right\vert _{\infty }\right\} \\
&\leq \frac{1}{1+r\Delta t}\left\vert \left\vert V^{n}\right\vert
\right\vert _{\infty }
\end{aligned}
\end{equation}
% \begin{eqnarray*}
% \left\vert \left\vert V^{n+1}\right\vert \right\vert _{\infty }
% &=&\left\vert V_{i_{0}}^{n+1}\right\vert   \notag \\
% &\leq &\frac{1}{1+r\Delta t}\left\{ \frac{\Delta t}{2h}\left( \frac{\Sigma
% ^{2}\left( \mathbbm{I}_{h}V_{i}^{n}\right) }{h}+\frac{\Sigma ^{2}\left( %
% \mathbbm{I}_{h}V_{i}^{n}\right) }{2}-r\right) \left\vert \left\vert
% V^{n}\right\vert \right\vert _{\infty }\right.   \notag \\
% &&\left. +\left( 1-\frac{\Sigma ^{2}\left( \mathbbm{I}_{h}V_{i}^{n}\right)
% \Delta t}{h^{2}}\right) \left\vert \left\vert V^{n}\right\vert \right\vert
% _{\infty }+\frac{\Delta t}{2h}\left( r+\frac{\Sigma ^{2}\left( \mathbbm{I}%
% _{h}V_{i}^{n}\right) }{h}-\frac{\Sigma ^{2}\left( \mathbbm{I}%
% _{h}V_{i}^{n}\right) }{2}\right) \left\vert \left\vert V^{n}\right\vert
% \right\vert _{\infty }\right\}   \notag \\
% &\leq &\frac{1}{1+r\Delta t}\left\vert \left\vert V^{n}\right\vert
% \right\vert _{\infty },
% \end{eqnarray*}%
which then results in
\begin{equation}
\left\vert \left\vert V^{n+1}\right\vert \right\vert _{\infty }\leq \frac{1}{%
\left( 1+r\Delta t\right) ^{N}}\left\vert \left\vert \phi \right\vert
\right\vert _{\infty }.  \label{sta2}
\end{equation}%
Combining equations \eqref{sta1} and \eqref{sta2} gives%
\begin{equation*}
\underset{n}{\max }\left\vert \left\vert V^{n}\right\vert \right\vert
_{\infty }\leq \max \left( \left\vert \left\vert \phi \right\vert
\right\vert _{\infty },\underset{n}{\max }\left\vert \varphi ^{n}\right\vert
\right) .
\end{equation*}
\end{proof}
\begin{theorem}
\textrm{(Convergence to the viscosity solution) }
Suppose the time step $\Delta t$ and space step $h$ satisfy the stability condition
$$
\overline{\Sigma }\sqrt{\Delta t}\leq h\leq \frac{2\underline{\Sigma }^{2}}{\max \left( 2r-\underline{\Sigma }^{2},\overline{\Sigma }^{2}-2r\right) },
$$
provided that $\max \left( 2r-\underline{\Sigma }^{2},\overline{\Sigma }^{2}-2r\right) > 0$.
Then, the numerical solution generated by the explicit scheme \eqref{explicit1} converges to the viscosity solution of the nonlinear PDE \eqref{G-option1} as $h \to 0$.
\label{Convergence1}
\end{theorem}

\begin{proof}
Since the scheme \eqref{explicit1} is consistent, $%
l_{\infty }-$stable, and monotone, the convergence follows from the results of Barles and Souganidis \cite{BS2} directly.
\end{proof}
%\begin{definition}
%\textrm{(Monotonicity)} \label{def-mono} The scheme \eqref{explicit1} is
%monotone if for all $i$,
%\begin{equation}
%g_{i}\left( U_{i}^{n}+\epsilon _{i}^{n},U_{i}^{n+1},\{U_{k}^{n+1}\}_{k\in
%N_{i}}\right) \leq g_{i}\left( U_{i}^{n},U_{i}^{n+1},\{U_{k}^{n+1}\}_{k\in
%N_{i}}\right) ,\ \ \forall \epsilon _{i}^{n}\geq 0,  \label{g-diag}
%\end{equation}%
%and
%\begin{equation}
%g_{i}\left( U_{i}^{n},U_{i}^{n+1}+\epsilon _{i}^{n+1},\{U_{k}^{n+1}+\epsilon
%_{k}^{n+1}\}_{k\in N_{i}}\right) \geq g_{i}\left(
%U_{i}^{n},U_{i}^{n+1},\{U_{k}^{n+1}\}_{k\in N_{i}}\right) ,\ \forall
%\epsilon _{i}^{n+1},\epsilon _{k}^{n+1}\geq 0.  \label{g-off-diag}
%\end{equation}
%\end{definition}
\begin{remark}
An explicit scheme for equation \eqref{G-op1} reads as%
\begin{equation*}
\left\{
\begin{array}{l}
\delta _{t}U_{i}^{n+1}+rS\delta _{h}U_{i}^{n+1}+\frac{1}{2}S^{2}\Sigma
^{2}\left( \delta _{h}^{2}U_{i}^{n+1}\right) \delta
_{h}^{2}U_{i}^{n+1}-rU_{i}^{n}=0, \\
U_{i}^{N}=\phi \left( S_{i}\right) .%
\end{array}%
\right.
\end{equation*}
It is easy to obtain that the condition for the explicit scheme to converge to the viscosity solution of the nonlinear PDE \eqref{G-op1} is
\begin{equation}
S_{\max }\overline{\Sigma}\sqrt{\Delta t}\leq h_{s}\leq \frac{S_{\min }\underline{%
\Sigma }^{2}}{r},
\end{equation}%
where
\begin{equation*}
S_{\max }=e^{L}, S_{\min }=e^{-L},
h_{s}=\frac{S_{\max }-S_{\min }}{M}.
\end{equation*}%
Then, the constraint for $\Delta t$ is obtained as%
\begin{equation}
\frac{S_{\max }^{2}\bar{\Sigma}^{2}M^{2}}{\left( S_{\max }-S_{\min }\right)
^{2}}\Delta t\leq 1.  \label{con1}
\end{equation}%
In contrast, when applying the logarithmic transformation with a fixed number of spatial nodes, the convergence condition established in Theorem \ref{Convergence1} yields
\begin{equation}
\overline{\Sigma }\sqrt{\Delta t}\leq h_{x}\leq \frac{2\underline{\Sigma }%
^{2}}{\max \left( 2r-\underline{\Sigma }^{2},\overline{\Sigma }%
^{2}-2r\right) },
\end{equation}%
where
\begin{equation*}
h_{x}=\frac{\ln \left( S_{\max }\right) -\ln \left( S_{\min }\right) }{M}.
\end{equation*}%
Then, the constraint for $\Delta t$ is obtained as
\begin{equation}
\frac{\bar{\Sigma}^{2}M^{2}}{\left( \ln \left( S_{\max }\right) -\ln \left(
S_{\min }\right) \right) ^{2}}\Delta t\leq 1.  \label{con2}
\end{equation}
Combining inequalities \eqref{con1} and \eqref{con2}, we find that the grid
ratio (for time discretization) of the numerical scheme with logarithmic
transformation is $\frac{S_{\max }^{2}\left( \ln \left( S_{\max }\right)
-\ln \left( S_{\min }\right) \right) ^{2}}{\left( S_{\max }-S_{\min }\right)
^{2}}$ times lower than that without the transformation. For notational simplicity and when the context is unambiguous, we shall denote $h_x$ as $h$ in the subsequent analysis.
\label{remark}
\end{remark}
\begin{theorem} \label{rate1}
\textrm{(Rate of convergence)} Let $v$ be the viscosity solution of equation \eqref{G-option1}, $V$ be the numerical solution of equation \eqref{explicit1}.
If condition $\overline{\Sigma }\sqrt{\Delta t}\leq h\leq \frac{2\underline{\Sigma }^{2}}{\max \left( 2r-\underline{\Sigma }^{2},\overline{\Sigma }^{2}-2r\right) }$ is satisfied and there exists some $\beta \in \left( 0,1\right)$, such that $v\in C^{1+\beta /2,2+\beta }\left( \left[0,T\right] \times \Omega\right)$, then
\begin{equation}
\max_n\left\vert \left\vert v^{n}-V^{n}\right\vert \right\vert
_{\infty }\leq C\left( \Delta t^{\frac{\beta }{2}}+h^{\beta }\right), \end{equation}
where $C$ is a positive constant independent of $\Delta t$ and $h$.
\end{theorem}
\begin{proof}
Approximating the derivatives by corresponding difference quotients in %
\eqref{G-option1}, we obtain
\begin{equation}
0=\delta _{t}v_{i}^{n+1}+R_{t}^{n}-r\left( \delta _{h}v_{i}^{n}+\widetilde{%
R_{h}}^{n}\right) -\frac{1}{2}\sup\limits_{_{\Sigma \in \left[
\underline{\Sigma },\overline{\Sigma }\right] }}\Sigma ^{2}\cdot\left(
\delta _{h}^{2}v_{i}^{n}+R_{h}^{n}-\left( \delta _{h}v_{i}^{n}+\widetilde{%
R_{h}}^{n}\right) \right) .
\end{equation}%
Here we have the truncation error terms
\begin{equation*}
R_{t}^{n}=v_{t}^{n}-\delta _{t}v_{i}^{n+1}=O((\Delta t)^{\beta /2}),%
\widetilde{R_{h}}^{n}=v_{x}^{n}-\delta _{h}v_{i}^{n}=O(h^{1+\beta
}),R_{h}^{n}=v_{xx}^{n}-\delta _{h}^{2}v_{i}^{n}=O(h^{\beta }).
\end{equation*}%
Thanks to $\sup (f+g)\leq \sup f+\sup g$, we have
\begin{equation}
\delta _{t}v_{i}^{n+1}-r\delta _{h}v_{i}^{n}+rv_{i}^{n+1}-\frac{1}{2}%
\sup\limits_{_{\Sigma \in \left[ \underline{\Sigma },\overline{%
\Sigma }\right] }}\Sigma ^{2} \cdot \left( \delta _{h}^{2}v_{i}^{n}-\delta
_{h}v_{i}^{n}\right) \leq R_{up}^{n},  \label{u-up}
\end{equation}%
where
\begin{equation}
R_{up}^{n}=-R_{t}^{n}+r\widetilde{R_{h}}^{n}+\frac{1}{2}\sup\limits_{_{%
\Sigma \in \left[ \underline{\Sigma },\overline{\Sigma }\right]
}} \Sigma ^{2} \cdot \left( R_{h}^{n}-\widetilde{R_{h}}^{n}\right) .
\end{equation}%
By the fact $\sup (f+g)\geq \sup f+\inf g$, we have
\begin{equation}
\delta _{t}v_{i}^{n+1}-r\delta _{h}v_{i}^{n}+rv_{i}^{n+1}-\frac{1}{2}%
\sup\limits_{_{\Sigma \in \left[ \underline{\Sigma },\overline{%
\Sigma }\right] }}\Sigma ^{2} \cdot\left( \delta _{h}^{2}v_{i}^{n}-\delta
_{h}v_{i}^{n}\right) \geq R_{low}^{n},  \label{u-low}
\end{equation}%
where
\begin{equation}
R_{low}^{n}=-R_{t}^{n}+r\widetilde{R_{h}}^{n}+\frac{1}{2}\inf\limits_{_{%
\Sigma \in \left[ \underline{\Sigma },\overline{\Sigma }\right]
}}\Sigma ^{2} \cdot\left( R_{h}^{n}-\widetilde{R_{h}}^{n}\right).
\end{equation}

Set $W_{i}^{n}=v_{i}^{n}-V_{i}^{n}$, then we have, by the fact $\sup
(f-g)\geq \sup f-\sup g$, that, for $0<i<M,$
\begin{eqnarray}
&&\delta _{t}W_{i}^{n+1}-r\delta _{h}W_{i}^{n}+rW_{i}^{n+1}-\frac{1}{2}%
\sup\limits_{_{\Sigma \in \left[ \underline{\Sigma },\overline{%
\Sigma }\right] }}\Sigma ^{2} \cdot \left( \delta _{h}^{2}W_{i}^{n}-\delta
_{h}W_{i}^{n}\right)   \notag \\
&=&\delta _{t}(v_{i}^{n+1}-V_{i}^{n+1})-r\delta
_{h}(v_{i}^{n}-V_{i}^{n})+r(v_{i}^{n+1}-V_{i}^{n+1})  \notag \\
&&\ \ -\frac{1}{2}\sup\limits_{_{\Sigma \in \left[ \underline{\Sigma }%
,\overline{\Sigma }\right] }}\Sigma ^{2} \cdot \left( \delta
_{h}^{2}(v_{i}^{n}-V_{i}^{n})-\delta _{h}(v_{i}^{n}-V_{i}^{n})\right)
\notag \\
&\leq &\delta _{t}v_{i}^{n+1}-r\delta _{h}v_{i}^{n}+rv_{i}^{n+1}-\frac{1}{2}%
\sup\limits_{_{\Sigma \in \left[ \underline{\Sigma },\overline{%
\Sigma }\right] }}\Sigma ^{2} \cdot \left( \delta _{h}^{2}v_{i}^{n}-\delta
_{h}v_{i}^{n}\right)   \notag \\
&&\ \ -\left( \delta _{t}V_{i}^{n+1}-r\delta _{h}V_{i}^{n}+rV_{i}^{n+1}-%
\frac{1}{2}\sup\limits_{_{\Sigma \in \left[ \underline{\Sigma },%
\overline{\Sigma } \right] }}\Sigma ^{2} \cdot \left( \delta
_{h}^{2}V_{i}^{n}-\delta _{h}V_{i}^{n}\right) \right)   \notag \\
&\leq &R_{up}^{n},
\end{eqnarray}%
where we have used \eqref{explicit1} and \eqref{u-up}. Proceeding as in the
proof of Lemma \ref{Stability}, we have
\begin{equation*}
\Vert W^{n+1}\Vert _{\infty }\leq \frac{1}{1+r\Delta t}\Vert W^{n}\Vert
_{\infty }+\Vert R_{up}^{n}\Vert _{\infty },
\end{equation*}%
with the initial condition $W_{i}^{0}=0$, which then results in
\begin{equation}
\Vert v^{n+1}-V^{n+1}\Vert _{\infty }\leq N\Vert R_{up}^{n}\Vert _{\infty }.
\label{rate-up}
\end{equation}

Similarly, let $\widetilde{W_{i}}^{n+1}=V_{i}^{n+1}-v_{i}^{n+1}$, then from %
\eqref{u-low}, we have
\begin{equation}
\Vert v^{n+1}-V^{n+1}\Vert _{\infty }\leq N\Vert R_{low}^{n}\Vert _{\infty }.
\label{rate-low}
\end{equation}%
Hence, we finally have
\begin{equation*}
\max_{n}\Vert v^{n}-V^{n}\Vert _{\infty }\leq N\ast \max_{n}(\Vert
R_{up}^{n}\Vert _{\infty }+\Vert R_{low}^{n}\Vert _{\infty })\leq C((\Delta
t)^{\beta /2}+h^{\beta }).
\end{equation*}
\end{proof}
\subsection{Numerical analysis for Implicit method}
The implicit scheme \eqref{iteration} leads to a nonlinear algebraic system which must be solved by
an inner iteration at each time step. In this section, we first prove the convergence of the inner iteration and then check the properties of consistence, stability, and monotonicity.
\subsubsection{ Convergence of inner iteration}
\begin{proposition}
\textrm{(Maximum principle)}\label{MP} If condition $h\leq \frac{2\underline{%
\Sigma }^{2}}{\max \left( 2r-\underline{\Sigma }^{2},\overline{\Sigma }%
^{2}-2r\right) }$ is satisfied, and $\{V_{i}^{n}\}$ satisfies
\begin{equation}
\mathbbm{L}_{h}V_{i}^{n}\equiv \delta _{t}V_{i}^{n}-r\delta
_{h}V_{i}^{n}+rV_{i}^{n}-\frac{\Sigma ^{2}}{2} \cdot \left( \delta
_{h}^{2}V_{i}^{n}-\delta _{h}V_{i}^{n}\right) \geq 0(\leq 0),\ n=1,...,N,\
0<i<M,  \label{linear eq}
\end{equation}%
where $\Sigma \in \left\{ \underline{\Sigma },\overline{\Sigma }%
\right\} .$ Then the minimum (maximum) of $\{V_{i}^{n}\}$ can only be
achieved at the initial or boundary points, unless $\{V_{i}^{n}\}$ is
constant.
\end{proposition}

Re-formulate the right-hand side of system \eqref{linear eq} into an
operator form as
\begin{equation*}
\mathbbm{L}_{h}V^{n}=\left( \frac{1}{\Delta t}I-r\delta _{h}+r-\frac{\Sigma
^{2}}{2}\cdot\delta _{h}^{2}+\frac{\Sigma ^{2}}{2}\cdot\delta _{h}\right) V^{n}-\frac{1%
}{\Delta t}V^{n-1}\equiv AV^{n}-\frac{1}{\Delta t}V^{n-1}.
\end{equation*}%
It is easy to check that, $A=(a_{ij})$ is an M-matrix, that is, $%
a_{ii}>0,a_{ij}\leq 0$, for $i\not=j$, and $\sum_{j\not=i}|a_{ij}|<a_{ii}$.
Then the above proposition follows.
% \begin{assumption}
% \label{assume2}The constraint on the spatial step size is
% $h\leq \frac{2\underline{\Sigma }^{2}}{\max \left( 2r-%
% \underline{\Sigma }^{2},\overline{\Sigma }^{2}-2r\right) }$.
% \end{assumption}

\begin{theorem}
\label{CII2} \textrm{(Convergence of inner iteration)} If condition $h\leq
\frac{2\underline{\Sigma }^{2}}{\max \left( 2r-\underline{\Sigma }^{2},%
\overline{\Sigma }^{2}-2r\right) }$ is satisfied, then for any initial guess
$V^{n+1,0}$, the iterative sequence $\{V^{n+1,k}\}_{k>0}$ in %
\eqref{iteration} {iteration} is bounded and monotonically increasing, so
converges to the unique solution to \eqref{implicit1}.
\end{theorem}

\begin{proof}
Denote by $\overline{V}^{k}=V^{n+1,k}$ and let $W^{k}=\overline{V}^{k+1}-%
\overline{V}^{k},$ $k\geq 1,$ then $W_{i}^{k}$ satisfies the following
difference equation
\begin{equation}
\left\{
\begin{array}{l}
\frac{W_{i}^{k}}{\Delta t}-r\delta _{h}W_{i}^{k}+rW_{i}^{k}-\frac{\Sigma
^{2}\left( \delta _{h}^{2}\overline{V}_{i}^{k}-\delta _{h}\overline{V}%
_{i}^{k}\right) }{2}\left( \delta _{h}^{2}\overline{V_{i}}^{k+1}-\delta _{h}%
\overline{V}_{i}^{k+1}\right)\\ +\frac{\Sigma ^{2}\left( \delta _{h}^{2}%
\overline{V}_{i}^{k-1}-\delta _{h}\overline{V}_{i}^{k-1}\right) }{2}\left(
\delta _{h}^{2}\overline{V}_{i}^{k}-\delta _{h}\overline{V}_{i}^{k}\right)
=0. \\
\left. W_{i}^{k}\right\vert _{_{X_{i}=-L}}=0, \\
\left. W_{i}^{k}\right\vert _{_{X_{i}=L}}=0.%
\end{array}%
\right.   \label{mono1}
\end{equation}%
Since $\frac{\Sigma ^{2}\left( \delta _{h}^{2}\overline{V}_{i}^{k}-\delta
_{h}\overline{V}_{i}^{k}\right) -\Sigma ^{2}\left( \delta _{h}^{2}\overline{V%
}_{i}^{k-1}-\delta _{h}\overline{V}_{i}^{k-1}\right) }{2}\left( \delta
_{h}^{2}\overline{V}_{i}^{k}-\delta _{h}\overline{V}_{i}^{k}\right) \geq 0,$
we have%
\begin{equation*}
\left\{
\begin{array}{l}
\frac{W_{i}^{k}}{\Delta t}-r\delta _{h}W_{i}^{k}+rW_{i}^{k}-\frac{\Sigma
^{2}\left( \delta _{h}^{2}\overline{V}_{i}^{k}-\delta _{h}\overline{V}%
_{i}^{k}\right) }{2}\left( \delta _{h}^{2}W_{i}^{k}-\delta
_{h}W_{i}^{k}\right) \geq 0, \\
W_{i}^{k}|_{_{X_{i}=-L}}=0, \\
W_{i}^{k}|_{_{X_{i}=L}}=0.%
\end{array}%
\right.
\end{equation*}

From the maximum principle, we have $W_{i}^{k}\geq 0$, for $0<i<M$, that is,
$\overline{V}_{i}^{k+1}\geq \overline{V},k\geq 1$. So $\{\overline{V}%
^{k}\}_{k>0}$ is a monotonic increasing sequence. Now we check the
boundedness of the sequence. 
From Proposition \ref{MP}, the maximum principle holds for the system \eqref{iteration}. Furthermore, it is easy to prove that the maximum value cannot be attained on the left boundary.
It follows that
\begin{equation}
\left\vert \left\vert \overline{V}^{k+1}\right\vert \right\vert _{\infty
}\leq \max \left(\left\vert \left\vert \phi \right\vert \right\vert
_{\infty },\underset{n}{\max }\left\vert \varphi ^{n}\right\vert \right).
\end{equation}%
Consequently, as a monotonic and bounded sequence, $\overline{V}^{k}$
converges.
\end{proof}
\subsubsection{Monotonicity and Convergence of implicit scheme \texorpdfstring{\eqref{implicit1}}{implicit scheme}}

\begin{lemma}
\label{Consistency1} \textrm{(Consistency) }The implicit scheme \eqref{implicit1} is consistent.
\end{lemma}

\begin{proof}
The proof process is the same as that in Lemma \ref{Consistency}, and thus, it will not be reiterated here.
\end{proof}

\begin{lemma}
\label{lem-mono} \textrm{(Monotonicity)} If condition $h\leq \frac{2%
\underline{\Sigma }^{2}}{\max \left( 2r-\underline{\Sigma }^{2},\overline{%
\Sigma }^{2}-2r\right) }$ is satisfied, then the implicit scheme %
\eqref{implicit1} is monotone.
\end{lemma}

\begin{proof}
Proceeding as in the proof of Lemma \ref{mono11}, we first consider the perturbation
on $V_{i}^{n+1}$, and denote it by $\widetilde{V}_{i}^{n+1}=V_{i}^{n+1}+$ $%
\epsilon _{i}^{n},$ for $\epsilon _{i}^{n}\geq 0.$ We also denote by $%
\widetilde{V}_{k}^{n+1}=V_{k}^{n+1}+\epsilon _{k}^{n+1}$, for $\epsilon
_{k}^{n+1}\geq 0$ and $k\in N_{i}$. Then the difference between the two
sides of the inequality \eqref{g-diag} is
\begin{eqnarray*}
T&:= &g_{i}\left( \widetilde{V}_{i}^{n+1},\{\widetilde{V}_{k}^{n+1}\}_{k\in
N_{i}},V_{i}^{n}\right) -g_{i}\left( V_{i}^{n+1},\{V_{k}^{n+1}\}_{k\in
N_{i}},V_{i}^{n}\right)  \\
&=&\left( \frac{1}{\Delta t}+r\right) \epsilon _{i}^{n+1}-\frac{\widetilde{%
\Sigma }^{2}}{2} \cdot \left( \delta _{h}^{2}\widetilde{V}_{i}^{n+1}-\delta _{h}%
\widetilde{V}_{i}^{n+1}\right) +\frac{\widehat{\Sigma }^{2}}{2} \cdot \left( \delta
_{h}^{2}V_{i}^{n+1}-\delta _{h}V_{i}^{n+1}\right)  \\
&=&\left( \frac{1}{\Delta t}+r\right) \epsilon _{i}^{n+1}-\frac{\widetilde{%
\Sigma }^{2}}{2} \cdot \left( \delta _{h}^{2}\widetilde{V}_{i}^{n+1}-\delta _{h}%
\widetilde{V}_{i}^{n+1}\right) +\frac{\widetilde{\Sigma }^{2}}{2} \cdot \left(
\delta _{h}^{2}V_{i}^{n+1}-\delta _{h}V_{i}^{n+1}\right)  \\
&&+\frac{\widehat{\Sigma }^{2}}{2} \cdot \left( \delta _{h}^{2}V_{i}^{n+1}-\delta
_{h}V_{i}^{n+1}\right) -\frac{\widetilde{\Sigma }^{2}}{2} \cdot \left( \delta
_{h}^{2}V_{i}^{n+1}-\delta _{h}V_{i}^{n+1}\right)  \\
&\geq &\left( \frac{1}{\Delta t}+r+\frac{\Sigma ^{2}}{h^{2}}\right) \epsilon
_{i}^{n+1}\geq 0,
\end{eqnarray*}
since $\frac{\widetilde{\Sigma }^{2}}{2} \cdot \left( \delta _{h}^{2}\widetilde{V}%
_{i}^{n+1}-\delta _{h}\widetilde{V}_{i}^{n+1}\right) \geq \frac{\widehat{%
\Sigma }^{2}}{2} \cdot \left( \delta _{h}^{2}\widetilde{V}_{i}^{n+1}-\delta _{h}%
\widetilde{V}_{i}^{n+1}\right) ,$ where
\begin{equation*}
\widetilde{\Sigma }=\Sigma \left( \delta _{h}^{2}\widetilde{V}%
_{i}^{n+1}-\delta _{h}\widetilde{V}_{i}^{n+1}\right) ,\text{ }\widehat{%
\Sigma }=\Sigma \left( \delta _{h}^{2}V_{i}^{n+1}-\delta
_{h}V_{i}^{n+1}\right) .
\end{equation*}%
We now turn to the perturbation on $\{V_{k}^{n+1}\}_{k\in N_{i}}$ and $%
V_{i}^{n}.$ Then the difference between the two sides of the inequality %
\eqref{g-off-diag} is
\begin{eqnarray*}
T'&:= &g_{i}\left( V_{i}^{n+1},\{\widetilde{V}_{k}^{n+1}\}_{k\in N_{i}},%
\widetilde{V}_{i}^{n}\right) -g_{i}\left( V_{i}^{n+1},\{V_{k}^{n+1}\}_{k\in
N_{i}},V_{i}^{n}\right)   \notag \\
&=&-\frac{\epsilon _{i}^{n}}{\Delta t}-r\delta _{h}\widetilde{V}%
_{i}^{n+1}+r\delta _{h}V_{i}^{n+1}-\frac{\widetilde{\Sigma }^{2}}{2} \cdot \left(
\delta _{h}^{2}\widetilde{V}_{i}^{n+1}-\delta _{h}\widetilde{V}%
_{i}^{n+1}\right) +\frac{\widehat{\Sigma }^{2}}{2} \cdot \left( \delta
_{h}^{2}V_{i}^{n+1}-\delta _{h}V_{i}^{n+1}\right)   \notag \\
&=&-\frac{\epsilon _{i}^{n}}{\Delta t}-r\delta _{h}\widetilde{V}%
_{i}^{n+1}+r\delta _{h}V_{i}^{n+1}-\frac{\widetilde{\Sigma }^{2}}{2} \cdot \left(
\delta _{h}^{2}\widetilde{V}_{i}^{n+1}-\delta _{h}\widetilde{V}%
_{i}^{n+1}\right) +\frac{\widehat{\Sigma }^{2}}{2} \cdot \left( \delta _{h}^{2}%
\widetilde{V}_{i}^{n+1}-\delta _{h}\widetilde{V}_{i}^{n+1}\right)   \notag \\
&&+\frac{\widehat{\Sigma }^{2}}{2} \cdot \left( \delta _{h}^{2}V_{i}^{n+1}-\delta
_{h}V_{i}^{n+1}\right) -\frac{\widehat{\Sigma }^{2}}{2} \cdot \left( \delta _{h}^{2}%
\widetilde{V}_{i}^{n+1}-\delta _{h}\widetilde{V}_{i}^{n+1}\right)   \notag \\
&\leq &-\frac{1}{\Delta t}\epsilon _{i}^{n}-\left( \frac{r}{2h}+\frac{%
\widehat{\Sigma }^{2}}{2h^{2}}-\frac{\widehat{\Sigma }^{2}}{4h}\right)
\epsilon _{i+1}^{n+1}-\left( \frac{\widehat{\Sigma }^{2}}{2h^{2}}-\frac{r}{2h%
}+\frac{\widehat{\Sigma }^{2}}{4h}\right) \epsilon _{i-1}^{n+1}.
\end{eqnarray*}%
If condition $h\leq \frac{2\underline{\Sigma }^{2}}{\max \left( 2r-%
\underline{\Sigma }^{2},\overline{\Sigma }^{2}-2r\right) }$ is satisfied,
then we obtain $T'\leq 0.$ The monotonicity of scheme \eqref{explicit1} now
follows directly from Definition \ref{def-mono}.
\end{proof}

\begin{lemma}
\label{S} \textrm{(Stability)} If condition $h\leq \frac{2\underline{\Sigma }%
^{2}}{\max \left( 2r-\underline{\Sigma }^{2},\overline{\Sigma }%
^{2}-2r\right) }$ is satisfied, then the fully implicit scheme %
\eqref{implicit1} is stable, in the sense that
\begin{equation}
\underset{n}{\max }\left\vert \left\vert V^{n}\right\vert \right\vert
_{\infty }\leq \max \left( \left\vert \left\vert \phi \right\vert
\right\vert _{\infty },\underset{n}{\max }\left\vert \varphi ^{n}\right\vert
\right) .  \label{uncondition2}
\end{equation}
\end{lemma}

\begin{proof}
With the diagonal-dominated assumption, the maximum principle holds for
the system \eqref{implicit1}. Moreover, the maximum value cannot be achieved at the left boundary. The estimate on $l_{\infty }$ %
\eqref{uncondition2} follows directly.
\end{proof}

\begin{theorem}
\textrm{(Convergence to the viscosity solution) }If condition $h\leq \frac{2%
\underline{\Sigma }^{2}}{\max \left( 2r-\underline{\Sigma }^{2},\overline{%
\Sigma }^{2}-2r\right) }$ is satisfied, then the implicit scheme %
\eqref{implicit1} converges to the viscosity solution of the nonlinear PDE %
\eqref{G-option1}. \label{Convergence2}\
\end{theorem}

\begin{proof}
The proof process is the same as that in Theorem \ref{Convergence1}, and thus, it will not be reiterated here.
\end{proof}

\begin{theorem} \label{rate2}
\textrm{(Rate of convergence)} Let $v$ be the viscosity solution of equation \eqref{G-option1}, $V$ be the numerical solution of equation \eqref{implicit1}.
If condition $h\leq \frac{2\underline{\Sigma }%
^{2}}{\max \left( 2r-\underline{\Sigma }^{2},\overline{\Sigma }%
^{2}-2r\right) }$ is satisfied and there exists some $\beta \in \left( 0,1\right)$, such that $v\in C^{1+\beta /2,2+\beta }\left( \left[0,T\right] \times \Omega\right)$, then
\begin{equation}
\max_n\left\vert \left\vert v^{n}-V^{n}\right\vert \right\vert
_{\infty }\leq C\left( \Delta t^{\frac{\beta }{2}}+h^{\beta }\right), \end{equation}
where $C$ is a positive constant independent of $\Delta t$ and $h$.
\end{theorem}
\begin{proof}Rate of convergence follows using the
same steps as in Theorem \ref{rate1} and Pei et al. \cite{Pei}.
\end{proof}
\section{Numerical Examples}
{\label{sec:6}}
\subsection{Butterfly Spread}
To assess the performance of the proposed numerical methods, we first consider a butterfly spread, a standard benchmark for nonlinear option pricing models due to its non-convex payoff. The payoff function is defined as:
\begin{equation}
\phi(S) = \max(S - K_1, 0) - 2\max(S - K_m, 0) + \max(S - K_2, 0),
\end{equation}
where $K_m = (K_1 + K_2)/2$. 
Through the logarithmic transformation $X = \ln S$, we solve the PDE \eqref{G-option1} with the transformed initial condition:
\begin{equation}
    \phi(X) = \max(e^X - K_1, 0) - 2\max(e^X - K_m, 0) + \max(e^X - K_2, 0).
\end{equation}
For this test, we adopt the parameters from the benchmark problem in Pooley et al. \cite{Po}: $T=0.25$, $r=0.1$, $K_1=90$, $K_2=110$, $\underline{\Sigma}=0.15$, and $\overline{\Sigma}=0.25$. Our objective is to compute the option value at the at-the-money point $S_0 = 100$, which corresponds to $X = \ln(100)$ in the transformed domain.

% \textcolor{red}{Recall that $\tau = T - t$, so that the payoff is the value of the option at expiry $t = T$, or the initial condition of the PDE at $\tau = 0$.}
%Figure \ref{fig:Initial_Buttlefly} provides a diagram of a sample payoff function.
% Our objective is to compute the option value at $S=100$.

\begin{comment}
\begin{figure}[htbp]
    \centering
    \includegraphics[width=0.4\textwidth]{Initial_Butterfly.eps}
    \caption{Sample payoff function for a butterfly spread.}
    \label{fig:Initial_Buttlefly}
\end{figure}
\end{comment}

\begin{table}[h]
    \centering
    \caption{Convergence results for butterfly spread using the explicit finite difference scheme.}

    \label{tab:butterfly_explicit}
    \resizebox{0.8\textwidth}{!}{
    \begin{tabular}{@{}ccccc@{}}
    \toprule
  Timesteps & Spatial nodes  & $L^\infty$ error & Conv. rate & CPU time (s) \\
    \midrule
    16    & 161    & $1.196$ & --   & 0.0087 \\
    64    & 321     & 3.858$\times 10^{-1}$ & 1.63 &  0.0118 \\
    256   & 641    & 1.031$\times 10^{-1}$ & 1.90 & 0.0146 \\
    1024  & 1281   & 2.499$\times 10^{-2}$ & 2.04 & 0.0354 \\
    % 4096 & 2561    & 7.575$\times 10^{-3}$  & 1.72   & 0.1375 \\
    \bottomrule
    \end{tabular}}
\end{table}
%前四行

%比较 1.隐式无条件稳定 不满足现实网格比  精度高
%比较不做log变换的数值结果 空间定 不同的网格约束（变时间网格）

To observe the convergence of our numerical methods, we compute numerical solutions on a sequence of progressively refined grids.
The coarsest grid uses 161 uniform spatial points. At each subsequent refinement level, the number of spatial points is doubled, and the time step is quartered to maintain a consistent relationship.
The convergence rate is computed using the formula
\begin{equation}
\text{Convergence rate} = \frac{\ln(E_{\text{coarse}}/E_{\text{fine}})}{\ln(r)},
\end{equation}
where $E_{\text{coarse}}$ and $E_{\text{fine}}$ are the $L^\infty$ errors on consecutive grid levels, and $r=2$ is the refinement ratio for spatial discretization.
A high-resolution numerical solution computed on an exceptionally fine grid serves as the reference ("exact") solution for error estimation.
Table \ref{tab:butterfly_explicit} reports the convergence results for the explicit scheme, while Table \ref{tab:butterfly_implicit} summarizes the results for the implicit scheme. For the implicit method, the nonlinear system at each time step is solved using the iterative scheme \eqref{iteration} with a convergence tolerance of $10^{-6}$.
%A reference \textquotedblleft exact" solution is token as the numerical solution on a sufficiently fine grid. For both the explicit and implicit schemes, the timesteps is 16384, while the number of spatial nodes is taken as 5121 and 20481, respectively. \textcolor{red}{参考值}  Table \ref{tab:butterfly_explicit} presents the convergence results for the explicit scheme, while Table \ref{tab:butterfly_implicit} shows the results for the implicit scheme. For the implicit method, the nonlinear system at each time step is solved using the iterative scheme \eqref{iteration} with a convergence tolerance of $10^{-6}$.
The numerical results in Table \ref{tab:butterfly_implicit} further show that the convergence of the implicit scheme does not require the mesh ratio condition $\overline{\Sigma}\sqrt{\Delta t}\leq h$.
%In terms of the "reference accuracy", the implicit method achieves substantially higher precision than the explicit one, primarily because a much finer spatial discretization is employed.
%The numerical results presented in Table \ref{tab:butterfly_implicit} demonstrate that the convergence of the implicit scheme does not require the condition  of the mesh ratio $\overline{\Sigma }\sqrt{% \Delta t}\leq h$, and furthermore, the implicit scheme exhibits higher accuracy compared to the explicit scheme\textcolor{red}{就说“参考的精度”（之所以加引号是因为两只数值格式都有各自的参考值）隐式高很多 然后说是因为空间剖分了更多}.
\begin{table}[h]
    \centering
    \caption{Convergence results for butterfly spread using the implicit finite difference scheme.}
    \label{tab:butterfly_implicit}
    \resizebox{0.8\textwidth}{!}{
    \begin{tabular}{@{}ccccc@{}}
    \toprule
Timesteps & Spatial nodes  & $L^\infty$ error & Conv. rate & CPU time (s) \\
    \midrule
    16    & 641     & 6.032$\times 10^{-2}$ & --   & 0.2961 \\
    64    & 1281   & 1.565$\times 10^{-2}$ & 1.95 & 0.3198 \\
    256   & 2561   & 4.269$\times 10^{-3}$ & 1.87 & 0.4643 \\
    1024  & 5121    & 8.308$\times 10^{-4}$ & 2.36 & 1.9760 \\
    % 16384 & 20481  & 4.881540 & --                    & --   & 107.9618 \\
    \bottomrule
    \end{tabular}}
\end{table}

The convergence results in Tables \ref{tab:butterfly_explicit} and \ref{tab:butterfly_implicit} demonstrate that both schemes achieve approximately second-order accuracy. To evaluate the option value at the target point $X_0 = 2\ln(10)$ (corresponding to $S_0 = 100$), which does not coincide with grid points, we employ quadratic Lagrange interpolation. Table \ref{tab:butterfly_convergence_comparison} presents a comprehensive comparison of the convergence behavior for both explicit and implicit schemes.

\begin{table}[H]
    \centering
    \caption{Option values at $X_0 = 2\ln(10)$ using quadratic Lagrange interpolation for butterfly spread option.}
    \label{tab:butterfly_convergence_comparison}
    \resizebox{0.85\textwidth}{!}{
    \begin{tabular}{@{}cccccccc@{}}
    \toprule
    \multicolumn{4}{c}{\textbf{Explicit Scheme}} & \multicolumn{4}{c}{\textbf{Implicit Scheme}} \\
    \midrule
    Steps & Nodes & Value & Difference & Steps & Nodes & Value & Difference \\
    \midrule
    16    & 161   & 5.757737 & $8.762 \times 10^{-1}$ & 16    & 641   & 4.895220 & $1.372 \times 10^{-2}$ \\
    64    & 321   & 4.906351 & $2.485 \times 10^{-2}$ & 64    & 1281  & 4.883320 & $1.738 \times 10^{-3}$ \\
    256   & 641   & 4.879003 & $2.579 \times 10^{-3}$ & 256   & 2561  & 4.882627 & $1.045 \times 10^{-3}$ \\
    1024  & 1281  & 4.883390 & $1.808 \times 10^{-3}$ & 1024  & 5121  & 4.881922 & $3.400 \times 10^{-4}$ \\
    \bottomrule
    \end{tabular}}
    \begin{tablenotes}
        \footnotesize
        \item Reference value: 4.881582, computed via implicit scheme with 16384 temporal and 20481 spatial grid points.
    \end{tablenotes}
\end{table}

Figure \ref{fig:iter_Buttlefly} illustrates the convergence behavior of the inner iteration scheme for the implicit method. It can be seen that the number of iterations per time step is almost always only 2, demonstrating the high efficiency
of the implicit numerical algorithm. 

% These results validate the effectiveness of our iterative approach for solving the nonlinear system at each time step.

\begin{figure}[H]
    \centering
\includegraphics[scale=0.5]{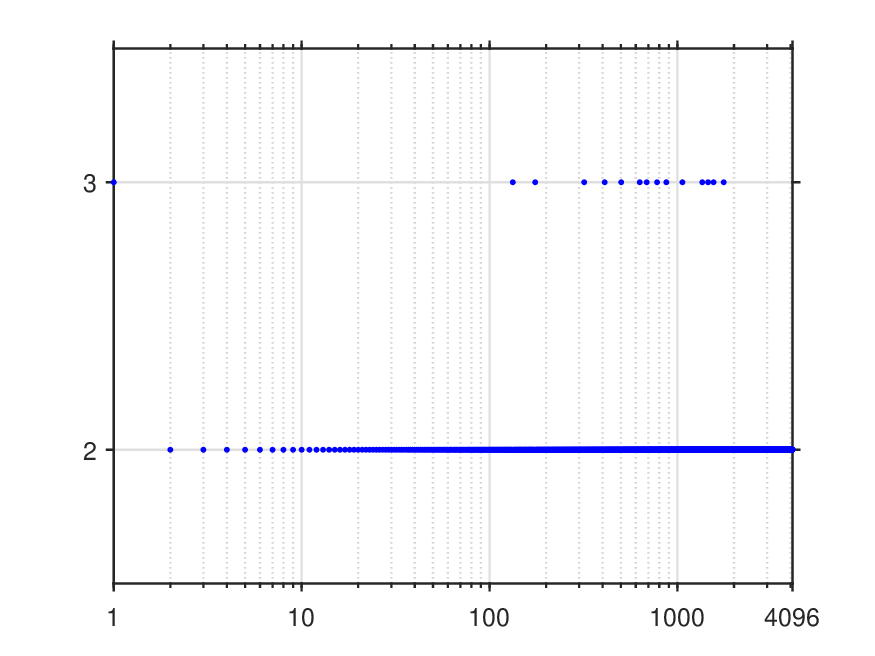}
\caption{The number of iterations within each time step.}
\label{fig:iter_Buttlefly}
\end{figure}

\subsection{Digital Call Option}
To further test the robustness of our schemes, we consider a digital call option, which presents a more significant challenge due to its discontinuous payoff at the strike price $K$:
%Digital options, also known as binary options, represent an extreme case of non-smoothness with a discontinuous terminal condition. 
\begin{equation}
U(T,S) =
\begin{cases}
1 & \text{if } S \geq K, \\
0 & \text{if } S < K,
\end{cases}
\end{equation}
The parameters for this test case are also adopted from Pooley et al. \cite{Po}, with $T=0.25$, $r=0.1$, $K=100$, $\underline{\Sigma}=0.15$, and $\overline{\Sigma}=0.25$.

\begin{comment}
\begin{figure}[htbp]
    \centering
    \includegraphics[width=0.4\textwidth]{initial_digital.eps}
    \caption{Sample payoff function for a digital call option.}
    \label{fig:initial_digital}
\end{figure}
\end{comment}
Through the logarithmic transformation $X = \ln S$, we solve the PDE \eqref{G-option1} with the transformed initial condition:
\begin{equation}
\phi(X) =
\begin{cases}
1 & \text{if } e^X \geq K, \\
0 & \text{if } e^X < K.
\end{cases}
\end{equation} 
%The objective is now to solve the solution $V(t=T,X=2\ln(10))$ of equation \eqref{G-option1}.
% For our numerical experiments, we set the parameters to $T=0.25$, $r=0.1$, $K=100$, $\underline{\Sigma}=0.15$, and $\overline{\Sigma}=0.25$.

\begin{table}[h]
    \centering
    \caption{Convergence results for the digital call option using the explicit finite difference scheme.}
    \label{tab:digital_explicit}
    \resizebox{0.8\textwidth}{!}{
    \begin{tabular}{@{}ccccc@{}}
    \toprule
    Timesteps & Spatial nodes  & $L^\infty$ error & Conv. rate & CPU time (s) \\
    \midrule
    64    & 961     & 4.142$\times 10^{-2}$ & --   & 0.0093 \\
    256   & 1921    & 1.989$\times 10^{-2}$ & 1.06 & 0.0210 \\
    1024  & 3841    & 9.209$\times 10^{-3}$ & 1.11 & 0.0599 \\
    4096  & 7681    & 3.908$\times 10^{-3}$ & 1.24 & 0.3246 \\
    % 16384 & 15361  & 0.691276 & --                    & --   & 2.4305 \\
    \bottomrule
    \end{tabular}}
\end{table}

\begin{table}[h]
    \centering
    \caption{Convergence results for the digital call option using the implicit finite difference scheme.}
    \label{tab:digital_implicit}
    \resizebox{0.8\textwidth}{!}{
    \begin{tabular}{@{}cccccc@{}}
    \toprule
    Timesteps & Spatial nodes & $L^\infty$ error & Conv. rate & CPU time (s) \\
    \midrule
    64    & 1281    & 1.195$\times 10^{-2}$ & --   & 0.0468 \\
    256   & 2561    & 2.238$\times 10^{-3}$ & 0.62 & 0.2089 \\
    1024  & 5121    & 5.483$\times 10^{-3}$ & 0.73 & 1.6998 \\
    4096  & 10241   & 1.836$\times 10^{-3}$ & 1.19 & 14.1552 \\
    % 16384 & 20481  & 0.690660 & --                    & --   & 109.4069 \\
    \bottomrule
    \end{tabular}}
\end{table}

%Tables \ref{tab:digital_explicit} and \ref{tab:digital_implicit} present the convergence results for the explicit and implicit schemes, respectively. As expected, the discontinuity in the payoff function poses significant challenges for both numerical methods. The convergence rates for both schemes are noticeably lower than those observed for the butterfly spread option.

Tables \ref{tab:digital_explicit} and \ref{tab:digital_implicit} present the convergence results for the explicit and implicit schemes, respectively. 
As anticipated, the discontinuity in the payoff function degrades the performance of both numerical schemes. 
The lack of regularity in the terminal condition reduces the global smoothness of the PDE's viscosity solution. Consequently, as shown in Tables \ref{tab:digital_explicit} and \ref{tab:digital_implicit}, the empirical convergence rates are markedly lower than those for the butterfly spread, approaching first-order accuracy.
To evaluate the option value at the target point $X_0 = 2\ln(10)$ (corresponding to $S_0 = 100$), we again employ quadratic Lagrange interpolation. A detailed comparison of the convergence behavior for both schemes is presented in Table \ref{tab:digital_convergence_comparison}.

\begin{table}[H]
    \centering
    \caption{Option values at $X_0 = 2\ln(10)$ using quadratic Lagrange interpolation for digital call option.}
    \label{tab:digital_convergence_comparison}
    \resizebox{0.85\textwidth}{!}{
    \begin{tabular}{@{}cccccccc@{}}
    \toprule
    \multicolumn{4}{c}{\textbf{Explicit Scheme}} & \multicolumn{4}{c}{\textbf{Implicit Scheme}} \\
    \midrule
    Steps & Nodes & Value & Difference & Steps & Nodes & Value & Difference \\
    \midrule
    64    & 961   & 0.652812 & $3.785 \times 10^{-2}$ & 64    & 1281  & 0.703109 & $1.245 \times 10^{-2}$ \\
    256   & 1921  & 0.673527 & $1.714 \times 10^{-2}$ & 256   & 2561  & 0.688523 & $2.139 \times 10^{-3}$ \\
    1024  & 3841  & 0.683729 & $6.933 \times 10^{-3}$ & 1024  & 5121  & 0.696152 & $5.490 \times 10^{-3}$ \\
    4096  & 7681  & 0.688773 & $1.889 \times 10^{-3}$ & 4096  & 10241 & 0.692500 & $1.838 \times 10^{-3}$ \\
    \bottomrule
    \end{tabular}}
    \begin{tablenotes}
        \footnotesize
        \item Reference value: 0.690662, computed via implicit scheme with 16384 temporal and 20481 spatial grid points.
    \end{tablenotes}
\end{table}

Figure \ref{fig:iter_digital} illustrates the convergence behavior of the inner iteration scheme for the implicit method, showing similar efficiency to the butterfly spread case with consistently low iteration counts per time step.

\begin{figure}[H]
    \centering
\includegraphics[scale=0.5]{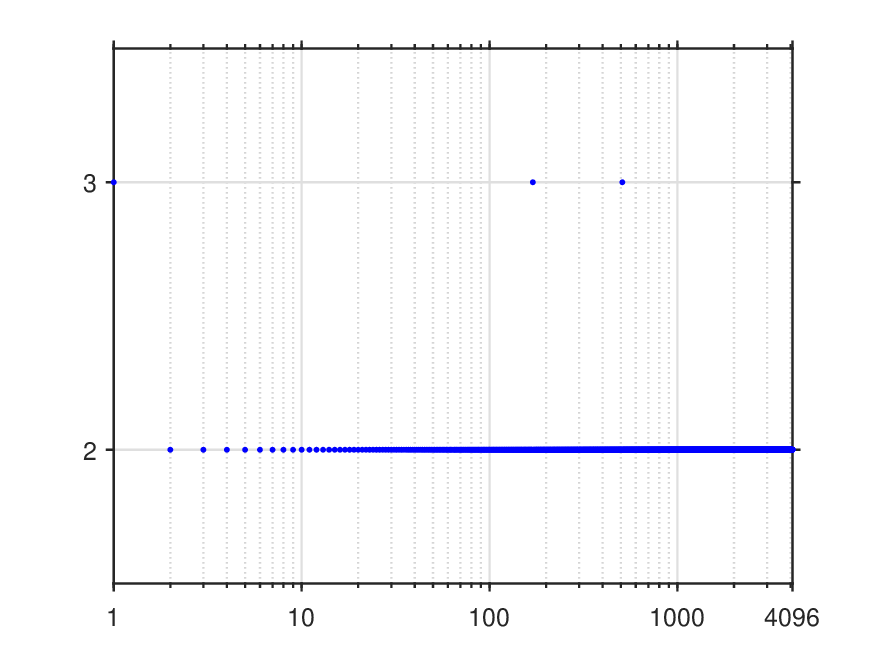}
\caption{The number of iterations within each time step.}
\label{fig:iter_digital}
\end{figure}

% These results validate the effectiveness of our iterative approach for solving the nonlinear system at each time step.

% \begin{figure}[H]
%     \centering
% \includegraphics[scale=0.45{fig/iter_digital.eps}
% \caption{The number of iterations within each time step.}
% \label{fig:iter_digital}
% \end{figure}

\subsection{Benefits of Logarithmic Transformation}

As detailed in Remark \ref{remark}, a key motivation for the logarithmic transformation $X = \ln S$ is to relax the restrictive stability constraint of the explicit finite difference scheme.
To quantify this theoretical advantage, we conduct a direct numerical comparison between solving the G-option equation in the original $S$-domain and the transformed $X$-domain.

\paragraph{Experimental Design:} To ensure a fair comparison, both formulations are solved using the same number of spatial grid points ($M=200, 400, 800$). The computational domain for the original formulation is set to $[S_{\min}, S_{\max}] = [50, 150]$, corresponding to $[X_{\min}, X_{\max}] = [\ln(50), \ln(150)]$ for the log-transformed problem. For each value of $M$, we determine the minimum number of time steps required by the stability condition for the explicit scheme in each domain.

Tables \ref{tab:butterfly_option_comparison} and \ref{tab:digital_option_comparison} present the numerical results for both schemes at different levels of grid refinement. The results provide quantitative evidence of the improvement achieved through the logarithmic transformation. The grid
ratio (for time discretization) of the numerical scheme with logarithmic transformation is approximately 2.7155 times lower than that without the transformation.

\begin{table}[H]
    \centering
    \caption{Comparison of numerical results for the butterfly spread option
    with and without logarithmic transformation.
    Reference value $V_\text{ref}=4.881540$ obtained by the implicit scheme
    with $16384$ temporal and $20481$ spatial steps (logarithmic transform).}
    \label{tab:butterfly_option_comparison}
    \resizebox{\textwidth}{!}{
    \begin{tabular}{@{}ccccccc@{}}
    \toprule
    & \multicolumn{2}{c}{$M=200$} & \multicolumn{2}{c}{$M=400$} & \multicolumn{2}{c}{$M=800$} \\
    \cmidrule(lr){2-3} \cmidrule(lr){4-5} \cmidrule(lr){6-7}
    & Without log & With log & Without log & With log & Without log & With log \\
    \midrule
    Spatial step size    & $5.00\times 10^{-1}$ & $5.49\times 10^{-3}$
                         & $2.50\times 10^{-1}$ & $2.75\times 10^{-3}$
                         & $1.25\times 10^{-1}$ & $1.37\times 10^{-3}$ \\
    Minimum time step    & $1.41\times 10^{3}$  & $5.18\times 10^{2}$
                         & $5.63\times 10^{3}$  & $2.07\times 10^{3}$
                         & $2.25\times 10^{4}$  & $8.29\times 10^{3}$ \\
    Numerical solution   & $4.88397$            & $4.88094$
                         & $4.88215$            & $4.88127$
                         & $4.88169$            & $4.88142$ \\
    Relative error       & $4.97\times 10^{-4}$ & $1.23\times 10^{-4}$
                         & $1.24\times 10^{-4}$ & $5.47\times 10^{-5}$
                         & $3.03\times 10^{-5}$ & $2.42\times 10^{-5}$ \\
    CPU time (s)         & $1.80\times 10^{-2}$ & $1.12\times 10^{-2}$
                         & $4.77\times 10^{-2}$ & $3.34\times 10^{-2}$
                         & $2.06\times 10^{-1}$ & $8.80\times 10^{-2}$ \\
    \bottomrule
    \end{tabular}
    }
\end{table}

\begin{table}[H]
    \centering
    \caption{Comparison of numerical results for the digital option with and without logarithmic
    transformation.
    Reference value $V_\text{ref}=0.690660$ obtained by the implicit scheme
    with logarithmic transformation using $16384$ temporal and $20481$
    spatial steps.}
    \label{tab:digital_option_comparison}
    \resizebox{\textwidth}{!}{
    \begin{tabular}{@{}ccccccc@{}}
    \toprule
    & \multicolumn{2}{c}{$M=200$} & \multicolumn{2}{c}{$M=400$} & \multicolumn{2}{c}{$M=800$} \\
    \cmidrule(lr){2-3} \cmidrule(lr){4-5} \cmidrule(lr){6-7}
    & Without log & With log & Without log & With log & Without log & With log \\
    \midrule
    Spatial step size    & $5.00\times 10^{-1}$ & $5.49\times 10^{-3}$
                         & $2.50\times 10^{-1}$ & $2.75\times 10^{-3}$
                         & $1.25\times 10^{-1}$ & $1.37\times 10^{-3}$ \\
    Minimum time step    & $1.41\times 10^{3}$  & $5.18\times 10^{2}$
                         & $5.63\times 10^{3}$  & $2.07\times 10^{3}$
                         & $2.25\times 10^{4}$  & $8.29\times 10^{3}$ \\
    Numerical solution   & $6.8029\times 10^{-1}$ & $6.8470\times 10^{-1}$
                         & $6.8514\times 10^{-1}$ & $6.8928\times 10^{-1}$
                         & $6.8755\times 10^{-1}$ & $6.9156\times 10^{-1}$ \\
    Relative error       & $1.50\times 10^{-2}$ & $8.63\times 10^{-3}$
                         & $7.99\times 10^{-3}$ & $2.00\times 10^{-3}$
                         & $4.50\times 10^{-3}$ & $1.36\times 10^{-3}$ \\
    CPU time (s)         & $3.86\times 10^{-2}$ & $1.19\times 10^{-2}$
                         & $7.08\times 10^{-2}$ & $3.17\times 10^{-2}$
                         & $2.03\times 10^{-1}$ & $8.59\times 10^{-2}$ \\
    \bottomrule
    \end{tabular}
    }
\end{table}

The numerical investigations underscore several critical advantages of employing the logarithmic transformation in the context of explicit difference schemes for G-option pricing equations:

\begin{compactitem}
\item \textbf{Relaxed Stability Constraint:} The transformation significantly loosens the stability condition inherent in explicit schemes. For instance, with $M=800$ spatial points, the minimum number of time steps required in the original $S$-domain is $22500$, whereas the transformed $X$-domain requires only $8286$—a reduction of approximately $63\%$.
\item \textbf{Improved Computational Efficiency:} Consequently, the log-transformed approach achieves comparable or superior accuracy at a substantially lower computational cost. As shown in the CPU time comparisons, the efficiency gains become more pronounced as the grid is refined.
\end{compactitem}

These results collectively establish the logarithmic transformation as a highly practical and theoretically justified preprocessing technique, which can be used for robust and efficient numerical processing of the G-option equation.

\section{Conclusions}
\label{sec:7} We have established the risk-neutral pricing problem of options under the G-expectation and developed explicit and implicit discretization schemes respectively to solve it. First, we describe the price of the underlying asset through the G-geometric Brownian motion, propose a risk-neutral pricing model for options under the G-expectation, and then derive the G-Black-Scholes equation. To improve computational efficiency and reduce numerical
cost, we introduce a logarithmic transformation of the asset price and obtain an alternative nonlinear PDE. Then, we demonstrate the consistency, stability, and monotonicity of the two numerical schemes respectively. In particular, for the implicit scheme, we prove the monotonic convergence of the nonlinear inner iteration at each time step. In addition, we provide an estimate of the convergence order. Finally, the numerical examples we present further confirm the accuracy and computational advantages of the proposed methods.

\section*{Acknowledgments}
This work of Yue was partially supported by the NSFC Grant 12371401. This work of Zheng was partially supported by the Postgraduate Research \& Practice Innovation Program of Jiangsu Province (Grant No. KYCX25 3502).

\end{document}